\documentclass[12pt,oneside,english]{article}
\usepackage{babel}
\usepackage[T1]{fontenc}
\usepackage[latin9]{inputenc}
\usepackage{simpler-wick}
\usepackage{geometry}
\usepackage{amstext,amsthm,amssymb}
\usepackage{bm}
\usepackage[justification=centering]{caption}
\usepackage{stmaryrd}
\usepackage{graphicx}
\usepackage{mathtools}
\usepackage{wasysym}
\usepackage{mathrsfs}
\usepackage{bbm}
\usepackage{esint}
\usepackage{dsfont}
\usepackage{xcolor}
\usepackage{verbatim}
\usepackage{subfigure}
\usepackage{setspace}
\usepackage{url}
\usepackage[percent]{overpic}
\usepackage{tikz}
	\usetikzlibrary{calc}
	\usetikzlibrary{decorations.text}
\usepackage{relsize}
\usepackage{titlesec}
\usepackage{upgreek}
\usepackage{yhmath}
\usepackage[hidelinks]{hyperref}

\makeatletter

\numberwithin{equation}{section}
\numberwithin{figure}{section}
\newtheoremstyle{theorem} 
	{\topsep}                    
	{\topsep}                    
	{\itshape}                   
	{}                           
	{\scshape\bfseries}                   
	{.}                          
	{.5em}                       
	{}  
\theoremstyle{theorem}
\newtheorem{thm}{Theorem}[section]

\newtheorem{prop}[thm]{Proposition}

\theoremstyle{definition}

\newtheoremstyle{remark} 
	{\topsep}                    
	{\topsep}                    
	{}                           
	{}                           
	{\scshape\bfseries}          
	{.}                          
	{.5em}                       
	{}  
\theoremstyle{remark}
\newtheorem{rmk}{Remark}[section]
\newtheorem{ex}{Example}[section]
\numberwithin{thm}{section}

\def\Z{{\mathbb{Z}}}

\def\R{{\mathbb{R}}}
\def\C{{\mathbb{C}}}
\def\Zpos{{\mathbb{Z}_{>0}}}
\def\Znneg{{\mathbb{Z}_{\geq0}}}
\def\UpperHalf{\mathbb{H}}
\def\im{\mathrm{Im}\,}

\newcommand{\black}{\color{black}}
\newcommand{\term}[1]{\textbf{#1}}
\newcommand{\mybar}[3]{
	\mathrlap{\hspace{#2}\overline{\scalebox{#1}[1]{\phantom{\ensuremath{#3}}}}}\ensuremath{#3}
}

\renewcommand{\Re}{\textnormal{Re}}
\renewcommand{\Im}{\textnormal{Im}}

\newcommand{\AssoAlg}{\mathcal{A}^{\textnormal{(chi)}}}
\newcommand{\FullAssoAlg}{\mathcal{A}}
\newcommand{\bk}{\mathbf{k}}
\newcommand{\bl}{\bm{\ell}}
\newcommand{\thectt}{\alpha}
\newcommand{\EnvSymFer}{\FullEnvSymFer^{\textnormal{(chi)}}}
\newcommand{\FullEnvSymFer}{\mathcal{SF}}

\newcommand{\SymFer}{\textnormal{\textbf{SF}}}
\newcommand{\SymFerBar}{\overline{\SymFer}}
\newcommand{\AlgEta}[1]{{\upeta}_{#1}}
\newcommand{\AlgChi}[1]{{\upchi}_{#1}}
\newcommand{\AlgEtaBar}[1]{\mybar{0.97}{-0.6pt}{\upeta}_{#1}}
\newcommand{\AlgChiBar}[1]{\mybar{1}{-0.6pt}{\upchi}_{#1}}
\newcommand{\AlgK}{\mathsf{k}}

\newcommand{\UniEnve}{\mathcal{U}}
\newcommand{\ChiralFock}{\FullFock^{\textnormal{(chi)}}}
\newcommand{\FullFock}{\mathcal{F}}
\newcommand{\GroundPartner}{\boldsymbol{\omega}}
\newcommand{\Ground}{\boldsymbol{\mathbbm{1}}}
\newcommand{\GroundEta}{\boldsymbol{\theta}}
\newcommand{\GroundChi}{\boldsymbol{\xi}}
\newcommand{\HolCurrentEta}{\boldsymbol{\eta}}
\newcommand{\HolCurrentChi}{\boldsymbol{\chi}}
\newcommand{\AntiHolCurrentEta}{\overline{\boldsymbol{\eta}}}
\newcommand{\AntiHolCurrentChi}{\overline{\boldsymbol{\chi}}}
\newcommand{\AlgL}[1]{{\mathbf{L}}_{#1}}
\newcommand{\AlgLBar}[1]{\mybar{0.85}{-1pt}{\mathbf{L}}_{#1}}
\newcommand{\AlgC}{\mathbf{C}}
\newcommand{\Vir}{\textnormal{\textbf{Vir}}}
\newcommand{\VirBar}{\mybar{1.02}{-0.6pt}{\textnormal{\textbf{Vir}}}}
\newcommand{\id}{\textnormal{id}}
\newcommand{\SugL}[1]{{\mathsf{L}}_{#1}}
\newcommand{\SugLBar}[1]{\mybar{0.6}{-0.3pt}{\mathsf{L}}_{#1}}
\newcommand{\EnvSugVir}{{\mathcal{V}ir}}
\newcommand{\End}{\textnormal{End}}
\newcommand{\BosChiralFock}{\ChiralFock_{\textnormal{bos}}}
\newcommand{\FerChiralFock}{\ChiralFock_{\textnormal{fer}}}
\newcommand{\spn}{\mathrm{span}}
\newcommand{\Verma}{\mathcal{V}}
\newcommand{\HiWe}{\mathcal{H}}
\newcommand{\Staggered}{\mathcal{S}}
\newcommand{\BosFullFock}{\FullFock_{\textnormal{bos}}}
\newcommand{\FerFullFock}{\FullFock_{\textnormal{fer}}}

\newcommand{\dist}{\mathsf{d}}
\newcommand{\ii}{\mathbbm{i}}
\newcommand{\domain}{\Omega}
\newcommand{\bigCorrFun}[3]{\big\langle #3 \big\rangle_{#1; #2}}
\newcommand{\BigCorrFun}[3]{\Big\langle #3 \Big\rangle_{#1; #2}}

\newcommand{\AnalyFun}[2]{C^{\,\omega}\big( #1 ; #2 \big)}
\newcommand{\Conf}[2]{\textnormal{Conf}_{#1}( #2 )}

\newcommand{\Green}{\mathbf{G}}

\newcommand{\sgntr}{\textnormal{sgn}}
\newcommand{\bdry}{\partial}
\newcommand{\HarmOfGreen}{\mathbf{g}}
\newcommand{\susyfactor}{\textnormal{C}}
\newcommand{\parity}{\mathsf{p}}

\newcommand{\widesim}[2][1.5]{
	\mathrel{\overset{#2}{\scalebox{#1}[1]{$\sim$}}}
}
\newcommand{\OPE}{\widesim[1.5]{}}

\newcommand{\field}{\boldsymbol{\varphi}}

\newcommand{\dd}{\textnormal{d}}

\newcommand{\Bij}{\mathfrak{B}}

\newcommand{\no}[1]{\, \rotatebox[]{90}
	{\scalebox{.7}{$\ \bullet\,\bullet$}}\,#1\,\rotatebox[]{90}{\scalebox{.7}{$\ \bullet\,\bullet$}}\,}

\titleformat{\subsection}[runin]
{\normalfont\bfseries}{\thesubsection.}{0pt}{\hspace{0.5em}}[.]
\titleformat{\section}
{\normalfont\Large\bfseries}{\thesection}{0pt}{\hspace{1em}}

\usepackage{enumitem}
\setlist[description]{leftmargin=4.3cm,labelindent=1cm,rightmargin=2.5cm}

\begin{document}

\title{\vspace{1cm} \Large\scshape\bfseries
Symplectic fermions in general domains
}

\setcounter{footnote}{1}
\author{{David Adame\hspace{0.9pt}-Carrillo\footnote{\texttt{david.adame-carrillo@helsinki.fi}}}\vspace{0.1cm}}


\date{\textit{\normalsize Department of Mathematics and Statistics, University of Helsinki }}

\maketitle

\vspace{30pt}

\begin{abstract}
We review the basic features of a logarithmic conformal field theory that arise in the context of the scaling limit of lattice models.
The theory of interest is the symplectic fermions, whose central charge is $-2$.
We provide an explicit construction of its space of fields as a logarithmic Fock space, and discuss its logarithmic structure as a representation of the Virasoro algebra.
The construction of the correlation functions is presented following the ideas of the bootstrap approach.
The text aims to be accessible to readers with little or no expertise in conformal field theory.
\end{abstract}

\vfill

\pagenumbering{gobble}

\newpage
\tableofcontents
\pagenumbering{arabic}
\setcounter{page}{2}
\newpage

\titlespacing*{\paragraph}{0pt}{5pt}{0pt}

\section{Introduction}
\label{sec:intro}

\paragraph{Statistical mechanics and conformal field theory.}
~It is a well-established conjecture  
that the scaling limit of two-dimensional critical models of statistical mechanics can be described by two-dimensional conformal field theories (CFTs).
Yet, making precise mathematical sense of such a statement has proven to be challenging in most cases.
This conjecture has inspired many beautiful mathematical works on the emergence of conformally covariant objects in the scaling limit of probabilistic models.

The results in the literature that are more relevant for the present contribution are those regarding multipoint correlation functions.
One example are the works on the dimer model by Kenyon (\cite{Kenyon-conf_inv,Kenyon-GFF}),
in which the author proved the scaling limit of the dimer height function to be the gaussian free 
field~(GFF),
an object that is well-known to carry the free boson CFT in its corre\-lation kernels---see, e.g., \cite{KangMakarov-GFF_CFT}.
Another instance of such results are the series of outstanding works 
\linebreak[1]
(\cite{Hongler-thesis,CHI-spin_correlations,CHI-primary_field_correlations}) in which the correlation functions of all Ising primary fields with general boundary conditions were proven to converge to conformally covariant functions in the scaling limit.

A key feature of any two-dimensional CFT is that its space of fields enjoys a rich algebraic structure.
In particular, the set of local fields constitutes a representation of the Virasoro algebra, an infinite-dimensional Lie algebra that encodes the conformal symmetries of the theory.
The correlation functions of all the (infinitely many) fields are conjecturally expected to be recovered from the scaling limit of the corresponding discrete model too.

Recently (\cite{ABK-DGFF_free_boson}), and for the first time, the full structure of a CFT has been obtained as the scaling limit of a discrete model.
The results can be summarized as follows.
Following the ideas developed in~\cite{HKV-CFT_on_lattice}, the space of lattice local fields of the discrete gaussian free field (DGFF),
can be rendered a Virasoro representation.
Moreover, this space is in one-to-one correspondence with ---i.e., it is isomorphic to--- the space of local fields of the free boson CFT.
Via this correspondence, the (suitably renormalized) correlation functions of DGFF lattice local fields converge, in the scaling limit, to the CFT correlation functions of the corresponding free boson local fields.

\paragraph{Logarithmic CFT.}
~In certain critical lattice models, non-local observables give rise to logarithmic singularities in correlation functions.
The CFTs that appear in the scaling limit of these models are known as logarithmic conformal field theories (logCFTs).
Besides featuring fields that produce logarithmic singularities in correlation functions, the space of fields of a logCFT exhibits a more intricate algebraic structure than non-logarithmic CFTs.
See~\cite{CreutzigRidout-review} for a review on logCFT.

A noteworthy example of this phenomenon are crossing probabilities in critical percolation.
In fact, the CFT arguments that Cardy used in his original derivation of the celebrated Cardy--Smirnov formula (\cite{Cardy_percolation,Smirnov-percolation}) can be accommodated only 
in a logarithmic CFT;
see, e.g., \cite{MathieuRidout-percolation_logCFT}.
More recently, the understanding of percolation as a logCFT has seen further developments:
certain probabilistic observables have been identified as CFT fields that produce logarithmic dependencies within correlation functions---see~\cite{Camia-conf_conn_prob,CamiaFeng-log_percolation}.

A prominent example of a logarithmic CFT is the symplectic fermions.
The theory was first expounded in the full plane by Kausch in~\cite{Kausch-SF}, building on earlier works as~\cite{Guraire-log_op_CFT,Kausch-curiosities, GabKau-local}.
It was subsequently reformulated in the language of vertex operator algebras (VOA) in \cite{Abe-symfer_VOA}),
and has ever since been studied by the VOA community; see~\cite{DavydovRunkel_holo-symfer} and references therein.

Associated to every CFT there is a ---typically real--- parameter $c$ called \textit{central charge} that is encoded in the Virasoro structure of the space of fields.
In essence, the central charge determines how a system responds to the introduction of a macroscopic length scale (\cite[Section 5.4]{yellow_book}).
For the symplectic fermions, this parameter takes the value $-2$.
What makes this theory particularly appealing are its deep connections with several discrete probabilistic models that have long been predicted to fall within the universality class of CFTs at central charge $c=-2$.

\paragraph{Dimers, spanning trees, dense polymers and sandpiles.}
~In the recent work \cite{AdRu-fDGFF_to_symfer}, the \emph{fermionic discrete gaussian free field} (fDGFF) has been proven to converge in the scaling limit to the symplectic fermions CFT in the sense of \cite{ABK-DGFF_free_boson}.
This result is consistent with longstanding predictions that several discrete models should be described by (logarithmic) CFTs at $c=-2$, particularly in light of the various connections between the fDGFF and these models.
Motivated by these developments, the goal of this text is therefore twofold:
first, to provide a thorough review of the notions of symplectic fermions that arise in the analysis of this scaling-limit convergence; and second, to present these ideas in a way accessible to a broader audience, including readers with limited familiarity with conformal field theory.

The fDGFF naturally arises in tree and forest models ---see, e.g., \cite{CJSSS_trees-and-forests} or the more recent \cite{BCHS-fDGFF_forests}---, which can be exploited to compute correlation functions in general domains (\cite{CCRR_fGFF}).
Spanning trees have been classically linked to (log)CFT at $c=-2$;
see, e.g., \cite{MahieuRuelle-observables} where the connection goes via the abelian sandpile model or \cite{Saleur-polymers_CFT} in the context of dense polymers.
In more modern contributions like \cite{LPW-UST_CFT}, the uniform spanning tree has been linked with the symplectic fermions CFT in a random geometry setting.

The dimer model provides another key example.
It is widely regarded as a free fermionic theory that exhibits characteristic features of a conformally covariant theory with central charge $-2$ ---see \cite{IPRH-dimers_boundary_logCFT} and references therein---.
Furthermore, recent progress in \cite{Ada-discrete SF} has shown that certain observables in this model carry the algebraic structure of the symplectic fermions.

The abelian sandpile model is a particularly interesting case. 
It has been argued to be described by a logarithmic CFT at $c=-2$ that is \emph{not} equivalent to the symplectic fermions ---see \cite{Ruelle-review} and references therein---.
However, the two theories share some common fields, which has been exploited in \cite{CCRR_fGFF} to compute (the scaling limit of) the correlation functions of the height-one field in general domains of the complex plane.

Dense polymers (or spanning forests) provide yet another classical example of a critical model displaying features of a conformally covariant system with central charge $-2$ (\cite{Duplantier_dense-polymers}) that has been linked to the symplectic fermions CFT 
(\cite{Saleur-polymers_CFT,Ivashkevich-dense_polymers}).
A closely related model are the Hamiltonian walks, which exhibits too characteristics of a CFT at $c=-2$ (\cite{DuDa_hamiltonian-walks}).

\paragraph{Organization of the paper.}
~Section~\ref{sec: algebra} is devoted to the construction of the algebraic structures relevant to the symplectic fermions CFT,
that is, the symplectic fermions (asso\-ciative) algebra and its logarithmic Fock space representation.
We start the discussion from the \linebreak chiral version of these structures, in Sections~\ref{subsec: chiral module}--\ref{subsec: staggered}, in order to highlight the main algebraic features in a simpler setting; the full non-chiral module is then constructed in Section~\ref{subsec: non-chiral module}.
In Section~\ref{subsec: glimpse corr fun}, we provide some remarks on the algebraic structure of the space of fields that are relevant for the construction of correlation functions in later sections.
In particular, we observe that the space of fields possesses a family of (non-trivial) automorphisms which translates into an inherent ambiguity in correlation functions.
In Section~\ref{sec: corrfun}, the construction of the correlation functions of the theory is presented.
We start by motivating the choice of characterizing properties of the correlation functions in Sections~\ref{subsec: L-1}~and~\ref{subsec: OPE}.
Finally, we provide a characterization of the correlation functions in Section~\ref{subsec: characterization} along with some examples and remarks; its proof is provided in Section~\ref{subsec: proof}.

\vfill
\textbf{Acknowledgments:}
The author is grateful to Kalle Kyt\"ol\"a and Shinji Koshida for insightful discussions about the topics presented in this text.
The author thanks, too, Augustin Lafay, Philippe Ruelle, Wioletta Ruszel, and Dirk Schuricht for interesting discussions about topics related to this work;
and Leigh Foster, Vladimir Bo\v{s}kovi\'c, and Tom\'as Alcalde-L\'opez for valuable comments on this manuscript.
\newline
The author was supported by the Research Council of Finland (project 346309:~Finnish Centre of Excellence in Randomness and Structures,
"FiRST").
\newline
Part of this research was performed while the author was visiting the Institute for Pure and Applied Mathematics (IPAM), which is supported by the National Science Foundation (Grant No.~DMS-1925919).

\newpage
\section{The space of fields}
\label{sec: algebra}

In this section, we construct the space of fields of the symplectic fermions CFT.
In (bulk) two-dimensional conformal field theory, the space of fields is a representation of two commuting Virasoro algebras.
Naturally, this is also the case for the symplectic fermions CFT.

However, we start our discussion of the space of fields by constructing a chiral version of it, that is, a module of a single copy of the Virasoro algebra with a similar structure to the full non-chiral module.
Our motivation to do so is twofold.
Firstly, it is notationally lighter to work in the chiral setting, which allows us to discuss its logarithmic structure more transparently.
Secondly, by starting from the chiral module, we can appreciate a non-trivial difference with non-logarithmic CFT in the construction of the full non-chiral theory.
Gaberdiel and Kausch observed in~\cite{GabKau-local} that,
the non-chiral module of a logCFT needs to satisfy a condition that is trivially satisfied by the modules that appear in non-logarithmic CFT.
This fact renders the construction of the non-chiral module of a logCFT notably more involved.

Section~\ref{subsec: chiral module} contains the definition of the chiral logarithmic Fock space~$\ChiralFock$ of the symplectic fermions.
In Section~\ref{subsec: ground states and currents}, we single out the vectors within the module that play an important role in the construction of the symplectic fermions CFT throughout the article.
In Section~\ref{subsec: log action}, we outline how the logarithmic Fock space becomes a Virasoro represen\-tation through the Sugawara construction, and we highlight the logarithmic nature of it.
The space~$\ChiralFock$ can be proven to have, as a Virasoro subrepresentation, a staggered module---the proof of this fact and the necessary definitions are presented in Sections~\ref{subsec: hw reps} and~\ref{subsec: staggered}.
Finally, the full non-chiral space of fields of the symplectic fermions CFT is defined in Section~\ref{subsec: non-chiral module} as the (non-chiral) logarithmic Fock space of the symplectic fermions.
We conclude the section by giving some remarks in the direction of correlation functions in Section~\ref{subsec: glimpse corr fun}.

\subsection{The chiral module}
\label{subsec: chiral module}

Consider the vector space
\begin{align*}
\mathsf{V}_{\SymFer}
\coloneqq
\left( \bigoplus_{k\in\Z} \C \AlgEta{k} \right)
\oplus
\left( \bigoplus_{k\in\Z} \C \AlgChi{k} \right)\,
\end{align*}
and the associative algebra
\begin{align*}
\AssoAlg
\coloneqq
\bigoplus_{n\in\Znneg}
\big(
\mathsf{V}_{\SymFer}
\big)^{\otimes n}\,,
\end{align*}
where $\mathsf{V}^{\otimes n}$ denotes the tensor product of $n$ copies of $\mathsf{V}$.
In~$\AssoAlg$, we define the anticommutators by ${\{ v , u \} \coloneqq v \otimes u + u \otimes v}$.
The \term{chiral symplectic fermions associative algebra} is then defined as the quotient
\begin{align*}
\EnvSymFer
\coloneqq
\AssoAlg
\;\big/\;
\big(\,
\{\AlgEta{k},\AlgEta{\ell}\} ,
\{\AlgChi{k},\AlgChi{\ell}\} ,
\{\AlgEta{k},\AlgChi{\ell}\} - k \delta_{k+\ell}
\,:\,
k,\ell\in\Z
\,\big),
\end{align*}
of $\AssoAlg$ by a two-sided ideal generated by the symplectic fermions anticommutation relations.
That is, in~$\EnvSymFer$, we have
\begin{align*}
\big\{ \AlgEta{k} , \AlgChi{\ell} \big\}
\, = \, k \, \delta_{k+\ell} \,,
\mspace{50mu}
\big\{ \AlgEta{k} , \AlgEta{\ell} \big\}
=
\big\{ \AlgChi{k} , \AlgChi{\ell} \big\}
\, = \,
0 \,.
\end{align*}
To keep the notation light, we write~$v_1v_2\cdots v_n$ for the equivalence class of~${v_1 \otimes v_2 \otimes \cdots \otimes v_n}$ in the quotient~$\EnvSymFer$.

%

Alternatively, the algebra~$\EnvSymFer$ can be conceived as (a quotient of) the universal enve\-loping algebra of the Lie superalgebra of the symplectic fermions\footnote{
The Lie superalgebra is
$
\SymFer
\coloneqq
\left( \bigoplus_{k\in\Z} \C \AlgEta{k} \right)
\oplus
\left( \bigoplus_{k\in\Z} \C \AlgChi{k} \right)
\oplus
\C\AlgK\,,
$
with Lie superbrackets
\begin{align*}
\big\{ \AlgEta{k} , \AlgChi{\ell} \big\}
\, = \, k \, \delta_{k+\ell} \, \AlgK \,,
\mspace{50mu}
\big\{ \AlgEta{k} , \AlgEta{\ell} \big\}
=
\big\{ \AlgChi{k} , \AlgChi{\ell} \big\}
\, = \, 0 \, = \,
\big[ \AlgK , \SymFer \big]
\,.
\end{align*}
Then we have $\EnvSymFer \cong \UniEnve(\SymFer)/(\AlgK-1)$, where $\UniEnve(\SymFer)$ is the universal enveloping algebra of $\SymFer$.
}.
The Poincar\'e-Birkhoff-Witt~(PBW)~theorem states that~$\EnvSymFer$ admits the $\C$-vector space basis consisting of the set of finite words
\begin{align*}
	\AlgEta{k_r}\cdots\AlgEta{k_2}\AlgEta{k_1}\AlgChi{\ell_s}\cdots\AlgChi{\ell_2}\AlgChi{\ell_1}
\end{align*}
with $k_1 < k_2 < \cdots < k_r$ and $\ell_1 < \ell_2 < \cdots < \ell_s$.

The \term{chiral logarithmic Fock space}\footnote{The space can indeed be conceived as an alternating tensor algebra, hence it is a \textit{fermionic Fock space}. Besides, in Section~\ref{subsec: log action} we will see that it has a Virasoro structure that makes it suitable for \textit{logarithmic} CFT.} of the symplectic fermions is then the quotient
\begin{align*}
\ChiralFock
\, \coloneqq \,
\EnvSymFer
\, \big/ \,
\big(\, \AlgEta{k}\,,\; \AlgChi{k} \,\colon\, k>0  \,\big)
\end{align*}
of $\EnvSymFer$ by the left ideal generated by~$\AlgEta{k}$ and~$\AlgChi{k}$ with~$k>0$.
We conceive~$\ChiralFock$ as an \mbox{$\EnvSymFer$-module}, and we refer to $\AlgEta{k}$ and $\AlgChi{k}$ as the~\term{current modes} when we regard them as linear operators on~$\ChiralFock$.

We refer to the elements of~$\ChiralFock$ both as \term{states} or \term{fields}---this is a consequence of the \emph{field-state correspondence}, a key observation in the understanding of quantum field theories with conformal symmetries \cite{yellow_book}.
For any state~$\field\in\ChiralFock$, we let $(\EnvSymFer)\field$ denote the~sub\-module of $\ChiralFock$ cyclically generated by~$\field$.

The vector space~$\ChiralFock$ admits the PBW-type~basis 
\begin{align}\label{eq: chiral basis}
\AlgEta{-k_r}\cdots\AlgEta{-k_2}\AlgEta{-k_1}
\AlgChi{-\ell_s}\cdots\AlgChi{-\ell_2}\AlgChi{-\ell_1}
\big[\, 1 \,\big]
\end{align}
with~$r,s\in\Znneg$,
and the orderings
\mbox{${0 \leq k_1 < k_2 < \cdots < k_r}$} and \mbox{${0 \leq \ell_1 < \ell_2 < \cdots < \ell_s}$},
where \mbox{${[\,1\,]\in\ChiralFock}$} is the equivalence class of the vector~$1\in\EnvSymFer$.

The space~$\ChiralFock$ admits a natural $\Z_2$-grading in terms of the basis~\eqref{eq: chiral basis},
which we refer to as \term{parity} and is set to be $\boldsymbol{+}$~(plus) for the \term{bosonic states}
\begin{align*}
	\BosChiralFock
	\coloneqq
	\spn_\C
	\Big\{\,
	\AlgEta{-k_r}\cdots\AlgEta{-k_2}\AlgEta{-k_1}
	\AlgChi{-\ell_s}\cdots\AlgChi{-\ell_2}\AlgChi{-\ell_1}
	\big[\, 1 \,\big]
	\,\Big\vert\,
	s+r=0 \textnormal{ mod } 2
	\,\Big\}\,,
\end{align*}
and $\boldsymbol{-}$~(minus) for the \term{fermionic states}
\begin{align*}
	\FerChiralFock\,
	\coloneqq
	\spn_\C
	\Big\{\,
	\AlgEta{-k_r}\cdots\AlgEta{-k_2}\AlgEta{-k_1}
	\AlgChi{-\ell_s}\cdots\AlgChi{-\ell_2}\AlgChi{-\ell_1}
	\big[\, 1 \,\big]
	\,\Big\vert\,
	s+r=1 \textnormal{ mod } 2
	\,\Big\}\,.
\end{align*}

\subsection{Ground states and currents}
\label{subsec: ground states and currents}

We start the discussion of the structure of $\ChiralFock$ by defining the distinguished vectors that will play a fundamental role in the construction of the correlation functions in Section~\ref{sec: corrfun}.

There are four linearly independent vectors in $\ChiralFock$ that are annihilated by all positive current modes.
We refer to them as the \term{ground states} of $\ChiralFock$.
Those are the bosonic states
\begin{align*}
	\Ground & \coloneqq \big[\, \AlgChi{0}\AlgEta{0} \,\big]
	\mspace{10mu}\text{ and }
	\mspace{30mu}
	& (\text{\term{identity field}})
	\\
	\GroundPartner & \coloneqq \big[\, 1 \,\big],
	\mspace{30mu}
	& \phantom{\bigg\vert}(\text{\term{logarithmic partner} of the identity field})
\end{align*}
and the fermionic states
\begin{align*}
\mathllap{ \GroundChi \coloneqq -\big[\, \AlgChi{0} \,\big]
	\mspace{44mu}}
\text{ and }
\mathrlap{\mspace{42mu}
	\GroundEta \coloneqq -\big[\, \AlgEta{0} \,\big],}
\end{align*} 
where the square brackets~$[\,\cdot\,]$ denote the equivalence class in the quotient~$\ChiralFock$ of vectors in $\EnvSymFer$.
The names of $\GroundPartner$ and $\Ground$ will be justified later on in Remark~\ref{rmk: bosonic ground states} and Example~\ref{ex: identity}.
We refer to the fields~$\GroundChi$ and $\GroundEta$ as the \term{ground fermions}.

We note here that, as a representation of $\EnvSymFer$, the chiral logarithmic Fock space~$\ChiralFock$ is indecomposable \mbox{---$\GroundPartner$} is \mbox{cyclic---} yet reducible \mbox{---$\GroundPartner\notin(\EnvSymFer)\Ground$---}.

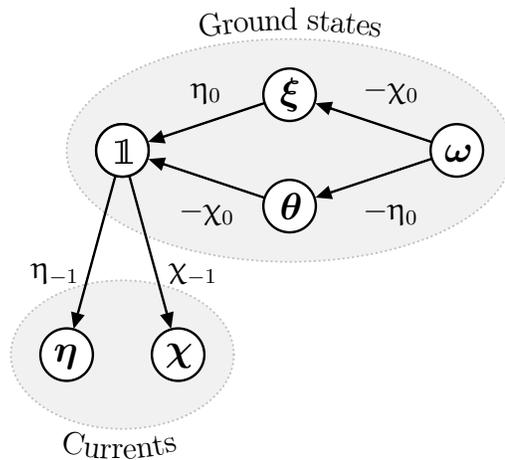
\begin{figure}[b!]
	\vspace{-5pt}
	\centering
	\begin{overpic}[scale=0.778]{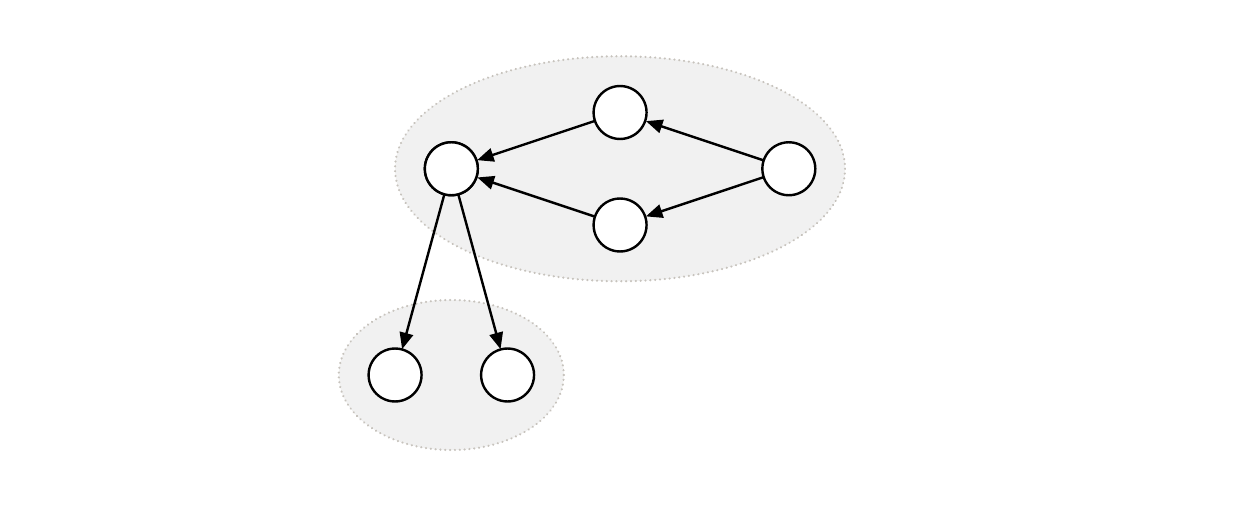}
		\put(49.2,31){\large $\GroundChi$}
		\put(49.2,21.6){\large $\GroundEta$}
		\put(35.45,26.15){\large $\Ground$}
		\put(62.6,26.3){\large $\GroundPartner$}
		\put(30.8,10){\large $\HolCurrentEta$}
		\put(39.9,10.1){\large $\HolCurrentChi$}
		\put(56,21.6){$-\AlgEta{0}$}
		\put(56,31.5){$-\AlgChi{0}$}
		\put(41,21.6){$-\AlgChi{0}$}
		\put(42,31.5){$\AlgEta{0}$}
		\put(29,16.8){$\AlgEta{-1}$}
		\put(40.2,16.8){$\AlgChi{-1}$}
		\begin{tikzpicture}
		\filldraw[white] (0,5) circle (2pt);
		\filldraw[white] (0,0) circle (2pt);
		\node (A) at (3.05,2.85) [] {}; 
		\node (B) at (13.15,2.85) [] {};
		\node (C) at (3.65,0) [] {}; 
		\node (D) at (8,0) [] {};
		\def\myshift#1{\raisebox{-2.5ex}}
		\draw [->,thick,color=white,postaction={decorate,decoration={text along path,text align=center,text={|\myshift|Ground states}}}] (A) to [bend left=45]  (B);
		\def\myshift#1{\raisebox{1ex}}
		\draw [->,thick,color=white, postaction={decorate,decoration={text along path,text align=center,text={Currents}}}]      (C) to [bend right=55] (D);
		\end{tikzpicture}
	\end{overpic}
	\centering
	\caption{
		The ground states, the currents, \\ and their relation via the action of the current modes.
	}
	\label{fig: ground states - chiral symalg}
\end{figure}

The fermionic fields
\begin{align*}
	\mathllap{\HolCurrentChi
		\coloneqq 
		\AlgChi{-1}\Ground
		=
		\big[\, \AlgChi{-1}\AlgChi{0}\AlgEta{0}\,\big]
		\mspace{30mu}}
	\text{ and }
	\mathrlap{\mspace{30mu}
		\HolCurrentEta
		\coloneqq
		\AlgEta{-1}\Ground
		=
		\big[\, \AlgEta{-1}\AlgChi{0}\AlgEta{0}\,\big]}
\end{align*}
play a special role in the construction of the symplectic fermions CFT.
In particular, within correlation functions, they carry the algebraic structure of the theory as their Fourier modes, which is then leveraged to bootstrap all the correlation functions of the theory.
They are called \term{currents}---this name is justified below in Remark~\ref{rmk: currents derivatives of fermions}.

In Figure~\ref{fig: ground states - chiral symalg}, we show visually the ground states and the currents, and their relations via current modes. 

\subsection{Logarithmic Virasoro action}
\label{subsec: log action}

The chiral logarithmic Fock space can be rendered a Virasoro representation via a Sugawara construction,
i.e.,
one can define Virasoro modes as 
an infinite sum quadratic on the current modes.
In the present case, the Virasoro action can be easily checked to be of the type that one expects in logarithmic CFT.

For $n\in\Z$, define the linear operator~$\ChiralFock\rightarrow\ChiralFock$ given by the formal sum
\begin{align}\label{eq: Sugawara}
	\SugL{n}
	\coloneqq
	\sum_{k \in \Z}
	\no{\AlgChi{n-k}\AlgEta{k}}\,,
\end{align}
where the normal ordering $\no{\cdots}$ is defined as follows:
\begin{align*}
	\no{\AlgChi{k}\AlgEta{\ell}}
	\coloneqq
	\begin{cases}
		\phantom{-} \AlgChi{k}\AlgEta{\ell} & \mspace{10mu} \textnormal{ if } \; k\leq \ell\,,
		\\
		\;- \AlgEta{\ell}\AlgChi{k} & \mspace{10mu} \textnormal{ if } \; k > \ell\,.
	\end{cases}
\end{align*}
In order to argue the well-posedness of these infinite sums, we need, first, the following observation about~$\ChiralFock$:
for any state~$\mathbf{v}\in\ChiralFock$, there exists $K\in\Znneg$ such that, for all $k \geq K$, we have
\begin{align*}
\AlgEta{k}\mathbf{v}
\;=\;
\AlgChi{k}\mathbf{v}
\;=\;
0\,.
\end{align*}
Well-definedness then follows since only finitely many terms produce non-zero vectors when acting on a given vector in $\ChiralFock$ with any $\SugL{n}$.

This definition leads to the following known result, referred to as the \term{Sugawara cons\-truction}, and justifies the name \term{Virasoro modes} for the operators~$\SugL{n}$.

\begin{rmk}\label{rmk: chiral sugawara}
	Let $[A,B]\coloneqq AB-BA$ denote the usual commutator of linear operators.
	The Virasoro modes~\eqref{eq: Sugawara} satisfy the Virasoro commutation relations with central charge~$-2$.
	That is, for all $n,m\in\Z$, we have
	\begin{align*}
	[\,\SugL{n},\SugL{m}\,]
	\, =  \;
	(n-m)\, \SugL{n+m}
	+
	\frac{c}{12}(n^3-n)\,\delta_{n+m}\,\id_{\ChiralFock}
	\end{align*}
	with $c=-2$.
	See \cite[Chapter 5.1]{Kac-VA beginners} for its discussion in the context of vertex algebras and \cite[Section 6.2]{Ada-discrete SF} for a detailed computational proof.
	\hfill$\diamond$ 
\end{rmk}

Our first observation regarding the Virasoro structure of $\ChiralFock$ is the non-diagonalizability of the Virasoro mode~$\SugL{0}$, a characteristic feature of the modules that appear in logarithmic CFT.
We record this fact in a remark for later reference.

\begin{rmk}\label{rmk: bosonic ground states}
	The Virasoro mode~$\SugL{0}$ acts on the bosonic ground states as
	\begin{align*}
	\mathllap{
		\SugL{0}\GroundPartner
		\,=\,
		\Ground
		\mspace{50mu}}
	\text{ and }
	\mathrlap{\mspace{50mu}
		\SugL{0}\Ground
		\,=\,
		0\,.}	
	\end{align*}
	It is clear from these that $\SugL{0}$ is non-diagonalizable in $\ChiralFock$.
	The first relation is the reason to call $\GroundPartner$ the \textit{logarithmic partner} of $\Ground$.
	\hfill$\diamond$
\end{rmk}

In spite of not being diagonalizable, the operator~$\SugL{0}$ provides a natural \mbox{$\Znneg$-grading} of $\ChiralFock$ via its decomposition into generalized eigenspaces.
The generalized $\SugL{0}$-eigenvalues are called \term{conformal dimensions}.
The grading can be expressed simply in terms of the basis~\eqref{eq: chiral basis};
a straightforward computation yields
\begin{align*}
\Big[\,
\SugL{0} -
\Delta_{\bk,\bl}\;\id_{\ChiralFock}
\Big]^2
\big(
\AlgEta{-k_r}\cdots\AlgEta{-k_2}\AlgEta{-k_1}
\AlgChi{-\ell_s}\cdots\AlgChi{-\ell_2}\AlgChi{-\ell_1}
\GroundPartner
\big)
\,=\, 0\,,
\end{align*}
where $\Delta_{\bk,\bl}\coloneqq\sum_{i=1}^rk_i+\sum_{j=1}^{s}\ell_j$ is the conformal dimension of the corresponding basis state.
In particular, this implies that $\SugL{0}$ has Jordan blocks of rank at most~$2$.

\begin{rmk}
\label{rmk: ground states annihilated by postive}
By a straightforward computation using the expression of the Virasoro modes~\eqref{eq: Sugawara} one can check that we have
\begin{align*}
	\mspace{170mu}
	\SugL{n}\Ground
	\,=\,
	\SugL{n}\GroundChi \,
	\,=\,
	\SugL{n}\GroundEta
	\,=\,
	\SugL{n}\GroundPartner
	\,=\,
	0
	\mspace{90mu}
	\text{ for }
	n > 0\,.
\end{align*}
In particular, the states~$\Ground$, $\GroundChi$, $\GroundEta$ and $\GroundPartner$ constitute a basis of the subspace with lowest conformal dimension~($0$) in~$\ChiralFock$, which is the reason to call them \emph{ground states}.
\hfill$\diamond$
\end{rmk}

Among the ground states, we have Virasoro \term{primary fields}, that is, $\SugL{0}$-eigenstates that are annihilated by all the positive Virasoro modes.

\begin{ex}
\label{ex: id is primary}
From Remarks~\ref{rmk: bosonic ground states} and \ref{rmk: ground states annihilated by postive}, it follows that the identity field~$\Ground$ is a primary field with conformal dimension~$0$.
\hfill$\diamond$
\end{ex}

\begin{ex}
Using the Sugawara construction~\eqref{eq: Sugawara}, one can check 
\begin{align*}
	\SugL{0}\GroundChi \,
	\,=\,
	\SugL{0}\GroundEta
	\,=\,
	0\,.
\end{align*}
From Remark~\ref{rmk: ground states annihilated by postive}, it then follows the fermionic fields~$\GroundChi$ and $\GroundEta$ are primary fields with conformal dimension~$0$.
\hfill$\diamond$
\end{ex}

However, those are not the only primary fields $\ChiralFock$.

\begin{ex}\label{ex: chiral currents}
	The currents~$\HolCurrentChi$ and $\HolCurrentEta$ can be seen to be Virasoro primary states by a straightforward computation.
	They are related to the ground states via the action of the Virasoro modes by
	\begin{align*}
	\mathllap{
		\SugL{-1}\GroundEta
		\,=\,
		\HolCurrentEta
		\mspace{50mu}}
	\text{ and }
	\mathrlap{\mspace{50mu}
		\SugL{-1}\GroundChi
		\,=\,
		\HolCurrentChi\,.}
	\end{align*}
	It is also straightforward to check that the currents have conformal dimension~$1$.
	In Section~\ref{subsec: L-1}, we will see that these algebraic relations have the interpretation~"$\HolCurrentEta = \partial\GroundEta$" and "$\HolCurrentChi = \partial\GroundChi$" within correlation functions. 
	For this reason, we interpret the states~$\HolCurrentEta$ and $\HolCurrentChi$ as the \textit{currents} of the ground fermions~$\GroundEta$ and $\GroundChi$, respectively.
	\hfill$\diamond$
\end{ex}
In Figure~\ref{fig: ground states - chiral Vir}, we show visually the ground states and the currents, and their relations via Virasoro modes.


\black
\newpage

\begin{figure}[t!]
	\vspace{-10pt}
	\centering
	\begin{overpic}[scale=0.778]{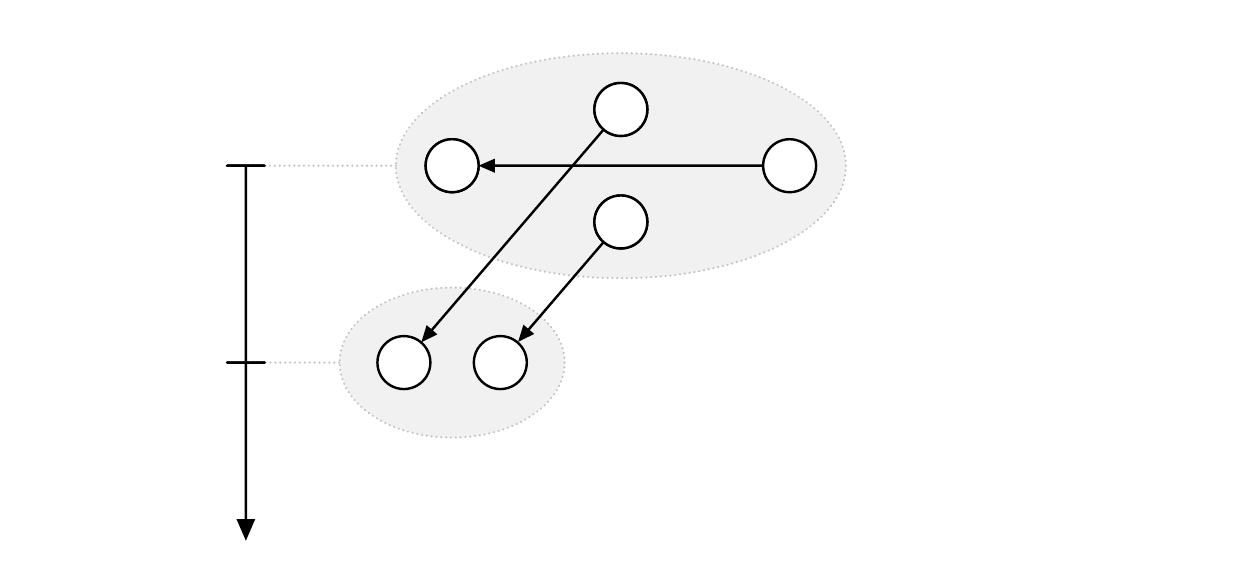}
		\put(49.2,37.15){\large $\GroundChi$}
		\put(49.2,27.8){\large $\GroundEta$}
		\put(35.5,32.35){\large $\Ground$}
		\put(62.65,32.55){\large $\GroundPartner$}
		\put(31.5,17){\large $\HolCurrentChi$}
		\put(39.4,17){\large $\HolCurrentEta$}
		\put(14.5,32.5){$0$}
		\put(14.6,16.7){$1$}
		\put(55.8,34.6){$\SugL{0}$}
		\put(34,24.6){$\SugL{-1}$}
		\put(45.7,20.9){$\SugL{-1}$}
		\put(8,7){\rotatebox{90}{conformal dimension}}
		\begin{tikzpicture}
			\filldraw[white] (0,5) circle (2pt);
			\filldraw[white] (0,-2.6) circle (2pt);
			\node (A) at (3.05,2.7) [] {}; 
			\node (B) at (13.15,2.7) [] {};
			\node (C) at (3.65,-0.05) [] {}; 
			\node (D) at (8,-0.05) [] {};
			\def\myshift#1{\raisebox{-2.5ex}}
			\draw [->,thick,color=white,postaction={decorate,decoration={text along path,text align=center,text={|\myshift|Ground states}}}] (A) to [bend left=45]  (B);
			\def\myshift#1{\raisebox{1ex}}
			\draw [->,thick,color=white, postaction={decorate,decoration={text along path,text align=center,text={Currents}}}]      (C) to [bend right=55] (D);
		\end{tikzpicture}
	\end{overpic}
	\centering
	\caption{
		The ground states, the currents,\\ and their relation via the Virasoro modes.
	}
	\label{fig: ground states - chiral Vir}
\end{figure}
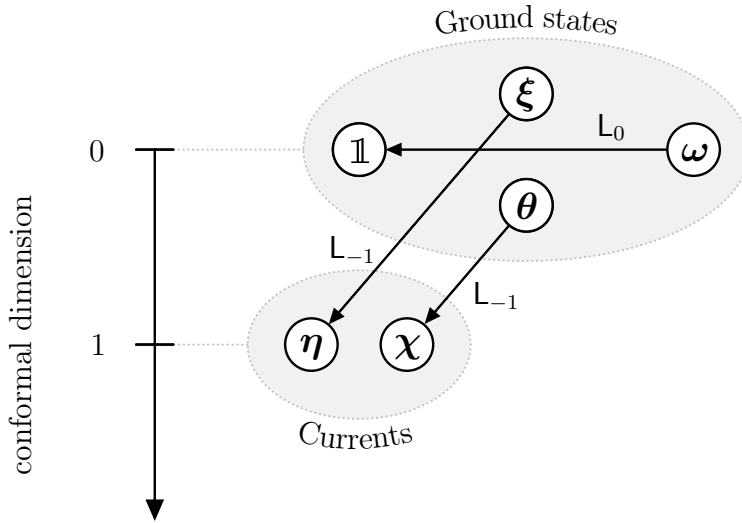

\subsection{Interlude: highest-weight representations at $c=-2$}
\label{subsec: hw reps}

Remark~\ref{rmk: bosonic ground states} in the previous section suggests that the chiral Fock space contains, as a submodule, a staggered module in the sense of Kyt\"ol\"a and Ridout~\cite{KR-staggered}.
That is the content of Proposition~\ref{prop: staggered} below, but before defining staggered modules, we recall the definition of highest-weight modules and their submodule structure at $c=-2$.
To that end, we review here some well-known facts about highest-weight modules that date back to the works of Feigin and Fuchs~\cite{FeiginFuchs_classification}, who built on a conjecture by Kac~\cite{Kac-conjecture} that they had proved~\cite{FeiginFuchs_conjecture}.

The \term{Virasoro algebra} is the infinite dimensional Lie algebra
\begin{align*}
	\Vir
	\coloneqq
	\bigoplus_{n\in\Z} \C\AlgL{n}\; \oplus \; \C\AlgC
\end{align*}
equipped with the Lie bracket
\begin{align*}
	\big[\,\AlgL{n}, \AlgL{m} \,\big]
	\;=\;
	(n-m) \, \AlgL{n+m} \, + \, \frac{n^3-n}{12} \, \delta_{n+m} \, \AlgC
	\mspace{50mu}
	\text{ and }
	\mspace{50mu}
	\big[\,\AlgL{n}, \AlgC \,\big] \;=\; 0\,.
\end{align*}

A Virasoro module~$\mathcal{M}$ is said to be a \term{highest-weight module} if $\AlgC$ acts as a constant multiple of the identity on $\mathcal{M}$, and it has a vector~${\mathbf{h}\in\mathcal{M}}$ that 
\begin{description}[style=sameline]
	\item[(cyclic)] is cyclic, i.e., any other vector in $\mathcal{M}$ can be obtained by the action of $\Vir$ on $\mathbf{h}$,
	\item[($\AlgL{0}$ e.v.)] is an eigenvector of $\AlgL{0}$, and
	\item[(primary)] is annihilated by all positive Virasoro modes, i.e
	\begin{align*}
		\mspace{20mu}
		\AlgL{n}\,\mathbf{h}=0
		\mspace{40mu}
		\text{ for } n>0\,.
	\end{align*}
\end{description}
The unique eigenvalue~$c\in\C$ of $\AlgC$ is called the \term{central charge} of $\mathcal{M}$ and the $\AlgL{0}$-eigen\-value~$h\in\C$ of $\mathbf{h}$ is called the \term{highest weight} of $\mathcal{M}$.
Note that $\AlgL{0}$ acts diagonally on any highest-weight module; its eigenvalues are called \term{conformal dimensions}.

Any highest-weight module is (isomorphic to) a quotient of a \term{Verma module}
\begin{align*}
	\Verma_h^{(c)}
	\coloneqq
	\UniEnve\big( \Vir \big) 
	/
	\big( 
	\AlgL{0} - h , \AlgL{1} , \AlgL{2} , \AlgC - c
	\big)\,,
\end{align*}
where $\UniEnve(\Vir)$ is the universal enveloping algebra of the Virasoro algebra,
$( X_1,\ldots,X_n )$ denotes the left ideal generated by \mbox{$X_1,\ldots,X_n\in\UniEnve(\Vir)$}, and \mbox{$h,c\in\C$} are complex numbers.
Note that, since $\AlgL{1}$ and $\AlgL{2}$ algebraically generate all the positive modes, we have $[\,\AlgL{n}\,]=0$ in the quotient $\Verma_h^{(c)}$ for all $n>0$.
By construction, since $[\,1\,]\in\Verma_h^{(c)}$ is a cyclic vector, all Verma modules are indecomposable.
However, they are not irreducible in general.

The relevant Verma module for the CFT at hand is $\Verma_0^{(-2)}$, which is, in fact, reducible.
Since the central charge is fixed to~$-2$ throughout the present work, we drop the superindex in the notation.
The submodule structure of $\Verma_0$ is commonly called of \emph{chain type}, which we record in the form of a remark for later reference.

\begin{rmk}\label{rmk: submodule structure of Verma0}
For the Verma module $\Verma_0$ of highest weight~$0$ and central charge~$c=-2$, we have the chain of submodules
	\begin{align*}
	\Verma_0
	\supset
	\Verma_1
	\supset
	\Verma_3
	\supset
	\cdots
	\supset\Verma_{k\choose 2}
	\supset\Verma_{{k+1}\choose 2}
	\supset
	\cdots\,.
	\end{align*}
That is, for each integer $k\in\Z_{\geq2}$ greater than $1$, the Verma module~$\Verma_0$ has a unique submodule isomorphic to $\Verma_{h_k}$ that is cyclically generated by a vector with conformal
dimension~${h_k}\coloneqq{k\choose2}$.
In turn, the submodule of highest weight~$h_k$ has a submodule isomorphic to $\Verma_{h_{k+1}}$ that is cyclically generated by a vector with conformal dimension~$h_{k+1}={{k+1}\choose2}$.
Those are, moreover, all the proper submodules of $\Verma_0$.
\hfill$\diamond$
\end{rmk}

Starting from $\Verma_0$, we can define, for $h\in{\Z_{\geq 2}\choose 2}$, the highest-weight representations
\begin{align}\label{eq: quotient hw reps}
	\HiWe^0_{h}
	\; \coloneqq \;
	\Verma_0 / \Verma_h
\end{align}
with highest weight~$0$.
Note that only $\HiWe^0_1$ is irreducible.

For concreteness, we write down the cyclic vectors of the two largest submodules of~$\Verma_0$ in the following remark, which will be necessary in the proof of Proposition~\ref{prop: staggered} below.

\begin{rmk}\label{rmk: singular vectors of Verma0}
Up to a non-zero multiplicative constant, the cyclic vectors of the sub\-modules~\mbox{$\Verma_1\subset\Verma_0$} and \mbox{$\Verma_3\subset\Verma_0$} are, respectively,
\begin{align*}
\AlgL{-1}\,[\,1\,]\,,
\mspace{50mu}
\text{ and }
\mspace{50mu}
\Big((\AlgL{-1})^2 - 2 \AlgL{-2}\Big)\,\AlgL{-1}
\,[\,1\,]\,.
\end{align*}
It takes a simple calculation to check that these two vectors are indeed annihilated by the action of the positive Virasoro modes in $\Verma_0$.	
\hfill$\diamond$
\end{rmk}

\subsection{Staggered (sub)module of $\ChiralFock$}
\label{subsec: staggered}

In this subsection, we introduce staggered mo\-dules following the definition of Kyt\"ol\"a and Ridout.
We, moreover, prove the chiral logarithmic Fock space~$\ChiralFock$ to have a staggered module as a submodule.

A Virasoro module~$\Staggered$ is a \term{staggered module} if $\AlgC$ acts as a constant multiple of the identity on $\Staggered$, and

\begin{description}[style=sameline]
	\item[(indecom.)] it is indecomposable,
	\item[($\AlgL{0}$ non-diag.)] the Virasoro mode~$\AlgL{0}$, as a linear operator $\Staggered\rightarrow\Staggered$, has Jordan blocks of rank 2, and
	\item[(sh.~ex.~seq.)] it satisfies the short exact sequence
	\begin{align*}
		0
		\,\longrightarrow\,
		\HiWe^{\text{L}}
		\,\longrightarrow\,
		\Staggered
		\,\longrightarrow\,
		\HiWe^{\text{R}}
		\,\longrightarrow\,
		0\,,
	\end{align*}
	where $\HiWe^{\text{L}}$ and $\HiWe^{\text{R}}$ are highest-weight modules.
\end{description}
The modules~$\HiWe^{\text{L}}$ and $\HiWe^{\text{R}}$ are called the \term{left module} and \term{right module} of $\Staggered$, respectively.
The unique eigenvalue of $c\in\C$ of $\AlgC$ is called \term{central charge}.
Note the central charges of $\Staggered$, $\HiWe^{\text{L}}$ and $\HiWe^{\text{R}}$ must coincide.

Provided a left and right module, the existence or uniqueness of the corresponding staggered module is not, a priori, guaranteed.
The answer to such question is rather non-trivial and was fully elaborated in~\cite{KR-staggered}.
Let us provide an example therefrom of a {staggered} module that {does} exist and whose short exact sequence implies its uniqueness.
We deli\-be\-rately choose to give as an example the module in~\cite[Example 2]{KR-staggered}, which we then prove to be a submodule of the chiral logarithmic Fock space in Proposition~\ref{prop: staggered} below.

\begin{ex}[{\cite[Example 2]{KR-staggered}}]\label{ex: staggered}
There exists a staggered module~$\Staggered^{\,0}_{1,3}$ with central charge~$c=-2$ and short exact sequence
\begin{align*}
0
\;\xrightarrow{\mspace{50mu}}\;
\HiWe^0_{1}
\;\xrightarrow{\mspace{50mu}}\;
\Staggered^{\,0}_{1,3}
\;\xrightarrow{\mspace{50mu}}\;
\HiWe^0_{3}
\;\xrightarrow{\mspace{50mu}}\;
0\,,
\end{align*}
where $\HiWe^0_{1}$ and $\HiWe^0_{3}$ are as defined in~\eqref{eq: quotient hw reps}.
The staggered module~$\Staggered^{\,0}_{1,3}$ is, moreover, unique up to Virasoro module isomorphism.
\hfill$\diamond$
\end{ex}

We now prove that $\ChiralFock$ contains $\Staggered^{\,0}_{1,3}$ as a Virasoro subrepresentation.
In the following \mbox{---and} for the rest of the \mbox{text---}, let \mbox{$\EnvSugVir\subset\End(\ChiralFock)$} denote the associative algebra generated by the Virasoro modes~$\SugL{n}$.
Then, for any state~$\field\in\ChiralFock$, we let $(\EnvSugVir)\field$ denote the submodule of $\ChiralFock$ cyclically generated by the state $\field$ under the action of the Virasoro modes.

\begin{prop}\label{prop: staggered}
We have the isomorphism
\begin{align*}
	(\EnvSugVir)\GroundPartner
	\;\cong\;
	\Staggered^0_{1,3}
\end{align*}
as Virasoro representations.
\end{prop}

\begin{proof}
We argued in Remark~\ref{rmk: bosonic ground states} that $(\EnvSugVir)\GroundPartner$ satisfies \textbf{($\SugL{0}$ non-diag.)}.
From the same remark, since we have $\Ground\in(\EnvSugVir)\GroundPartner$, we get that the module~$(\EnvSugVir)\Ground$ generated by $\Ground$ is a submodule of $(\EnvSugVir)\GroundPartner$.
In particular, we have the short exact sequence
\begin{align*}
	0
	\;\xrightarrow{\mspace{50mu}}\;
	(\EnvSugVir)\Ground
	\;\xrightarrow{\mspace{50mu}}\;
	\;(\EnvSugVir)\GroundPartner\;
	\;\xrightarrow{\mspace{50mu}}\;
	(\EnvSugVir)\GroundPartner/(\EnvSugVir)\Ground
	\;\xrightarrow{\mspace{46mu}}\;
	0\,.
\end{align*}
What is left in order to prove property \textbf{(sh.~ex.~seq.)}~is that the left module~$(\EnvSugVir)\Ground$ and the right module~$(\EnvSugVir)\GroundPartner/(\EnvSugVir)\Ground$ are isomorphic to highest-weight modules. In particular, we prove that they are isomorphic to $\HiWe_{0,1}$ and $\HiWe_{0,3}$ respectively.

\noindent
In the left module, the vector $\Ground\in(\EnvSugVir)\Ground$ satisfies \textbf{(cyclic)}~by construction, \textbf{($\AlgL{0}$ e.v.)}~by
\linebreak[2]
Remark~\ref{rmk: bosonic ground states}, and \textbf{(primary)}~by Remark~\ref{rmk: ground states annihilated by postive}.
The highest weight of $(\EnvSugVir)\Ground$ is~$0$.
A straight\-forward computation using the expression of the Virasoro modes~\eqref{eq: Sugawara} leads to
\begin{align*}
\SugL{-1}\Ground
=
0\,.
\end{align*}
From the submodule structure of $\Verma_0$ ---Remark~\ref{rmk: submodule structure of Verma0}---, it follows that the left module~$(\EnvSugVir)\Ground$ is isomorphic to $\HiWe_{0,1}$, which is irreducible.

\noindent
At this point we can easily see that $(\EnvSugVir)\GroundPartner$ satisfies \textbf{(indecom.)}~too:~by the commutation relations of the Virasoro algebra, the generator $\SugL{0}$ commutes ---up to a multiplicative constant--- with any other Virasoro generator.
Hence, we have $\SugL{0}(\EnvSugVir)\GroundPartner = (\EnvSugVir)\SugL{0}\GroundPartner = (\EnvSugVir)\Ground$.
It follows, then, that any submodule of $(\EnvSugVir)\GroundPartner$ must contain the irreducible $(\EnvSugVir)\Ground$; which implies that $(\EnvSugVir)\GroundPartner$ is indecomposable.

\noindent
As for the right module, we can argue that it is non-trivial since $\SugL{0}$ is diagonal in $(\EnvSugVir)\Ground$ but not in $(\EnvSugVir)\GroundPartner$.
In particular, we can see that $\GroundPartner\notin(\EnvSugVir)\Ground$ since the states~$\Ground$ and $\GroundPartner$ have the same conformal dimension.
It is also straightforward to check that the vector \mbox{$\big[\,\GroundPartner\,\big]\in(\EnvSugVir)\GroundPartner/(\EnvSugVir)\Ground$} satisfies \textbf{(cyclic)}~by construction, \textbf{($\AlgL{0}$ e.v.)}~by Remark~\ref{rmk: bosonic ground states}, and \textbf{(primary)}~by Re\-mark~\ref{rmk: ground states annihilated by postive}.
In view of Remark~\ref{rmk: singular vectors of Verma0}, in order to prove that $(\EnvSugVir)\GroundPartner/(\EnvSugVir)\Ground$ is isomorphic to $\HiWe_{0,3}$ it remains to prove 
\begin{align}\label{eq: right module}
\big[\,\SugL{-1}\GroundPartner\,\big]
\neq
\;0
\mspace{50mu}
\text{ and }
\mspace{50mu}
\big[\,
\big((\SugL{-1})^2 - 2 \SugL{-2}\big)\,\SugL{-1}\,\GroundPartner
\,\big]
=
\;0
\end{align}
in the quotient~$(\EnvSugVir)\GroundPartner/(\EnvSugVir)\Ground$.
As for the first of the two, a simple computation leads to
\begin{align*}
\SugL{-1}\GroundPartner
\, = \,
(\AlgEta{-1} \AlgChi{0} - \AlgChi{-1} \AlgEta{0})\GroundPartner
\,\neq \, 0\,,
\end{align*}
which has conformal dimension~$1$.
There is no vector in $(\EnvSugVir)\Ground$ with conformal dimension~$1$ ---since we have $\SugL{-1}\Ground=0$---, so we can conclude that $\SugL{-1}\GroundPartner\notin(\EnvSugVir)\Ground$,
that is $[\,\SugL{-1}\GroundPartner\,] \neq 0$.
As for the second equation in~\eqref{eq: right module}, a small computation leads to
\begin{align*}
\big((\SugL{-1})^2 - 2 \SugL{-2}\big)
\SugL{-1} 
\,\GroundPartner
\; = \;
-\, \SugL{-3}\Ground\,,
\end{align*}
that is, in the right module we have $[\,((\SugL{-1})^2 - 2 \SugL{-2})\,\SugL{-1}\,\GroundPartner\,] = 0$.
The submodule structure of~$\Verma_0$ ---Remark~\ref{rmk: submodule structure of Verma0}---, implies that the right module is isomorphic to $\HiWe_{0,3}$.

\noindent
By~Example~\ref{ex: staggered}, the exact short sequence of $\Staggered^{0}_{1,3}$ uniquely determines its isomorphism equi\-va\-lence class---we must have $(\EnvSugVir)\GroundPartner\,\cong\,\Staggered^0_{1,3}$.
\end{proof}

\subsection{The full non-chiral representation}
\label{subsec: non-chiral module}
In bulk two-dimensional CFT, the space of fields is a module of two commuting copies of the Virasoro algebra.
In~\cite{GabKau-local}, Gaberdiel and Kausch argued that the nilpotent part of the operator~$\AlgL{0}-\AlgLBar{0}$ should vanish in order for correlation functions to be single-valued (local).
Note that this condition becomes relevant only in logarithmic CFT, where the Virasoro modes~$\AlgL{0}$ and $\AlgLBar{0}$ are non-diagonalizable.
This purely non-chiral condition renders the construction of the full non-chiral module of the symplectic fermions a non-routine exercise; we deem appropriate to provide here the details.
The module we define is the one described in the original construction of the symplectic fermions CFT in the full plane in~\cite{Kausch-SF}.

\black
Similarly as in the chiral setting, we consider the vector space 
\begin{align*}
	\mathsf{V}_{\SymFer}
	\oplus
	\mathsf{V}_{\SymFerBar}
	\coloneqq
	\left( \bigoplus_{k\in\Z} \C\, \AlgEta{k} \right)
	\oplus
	\left( \bigoplus_{k\in\Z} \C\, \AlgChi{k} \right)
	\oplus
	\left( \bigoplus_{k\in\Z} \C\, \AlgEtaBar{k} \right)
	\oplus
	\left( \bigoplus_{k\in\Z} \C\, \AlgChiBar{k} \right)
\end{align*}
and the associative algebra
\begin{align*}
	\FullAssoAlg
	\coloneqq
	\bigoplus_{n\in\Znneg}
	\big(
	\mathsf{V}_{\SymFer}
	\oplus
	\mathsf{V}_{\SymFerBar}
	\big)^{\otimes n}\,.
\end{align*}
The full \term{(non-chiral) symplectic fermions associative algebra} is defined as the quotient
\begin{align*}
	\FullEnvSymFer
	\coloneqq
	\FullAssoAlg
	\;\big/\;
	\mathcal{I}
\end{align*}
where~$\mathcal{I}$ is the two-sided ideal generated by
\begin{align*}
	&
	\big\{ \AlgEta{k} , \AlgChi{\ell} \big\}
	\, - \, k \, \delta_{k+\ell} \,,
	\mspace{41mu}
	\big\{ \AlgEta{k} , \AlgEta{\ell} \big\}\,,
	\mspace{20mu}
	\big\{ \AlgChi{k} , \AlgChi{\ell} \big\}\,,
	&
	\text{(\term{holomorphic} copy)}
	\\
	\phantom{\Bigg\vert}
	&
	\big\{ \AlgEtaBar{k} , \AlgChiBar{\ell} \big\}
	\, - \, k \, \delta_{k+\ell} \,,
	\mspace{41mu}
	\big\{ \AlgEtaBar{k} , \AlgEtaBar{\ell} \big\}\,,
	\mspace{20mu}
	\big\{ \AlgChiBar{k} , \AlgChiBar{\ell} \big\}\,,
	&
	\mspace{60mu}
	\text{(\term{antiholomorphic} copy)}
	\\
	\text{ and }
	\mspace{20mu}
	&
	\big\{ \AlgChi{k} , \AlgChiBar{\ell} \big\}\,,
	\mspace{20mu}
	\big\{ \AlgChi{k} , \AlgEtaBar{\ell} \big\}\,,
	\mspace{20mu}
	\big\{ \AlgEta{k} , \AlgChiBar{\ell} \big\}\,,
	\mspace{20mu}
	\big\{ \AlgEta{k} , \AlgEtaBar{\ell} \big\}\,,
	&
	\text{(anticommuting copies)}
\end{align*}
for all~$k,\ell\in\Z$.

The \term{(non-chiral) logarithmic Fock space} of the  symplectic fermions is the quotient
\begin{align}\label{eq: full fock}
\FullFock
\, \coloneqq \,
\FullEnvSymFer
\, \big/ \,
\big(\, \AlgEta{0} - \AlgEtaBar{0}\,,\; \AlgChi{0} - \AlgChiBar{0}\,,\; \AlgEta{k}\,,\; \AlgChi{k}\,,\; \AlgEtaBar{k}\,,\; \AlgChiBar{k} \,\colon\, k>0  \,\big)
\end{align}
of the associative algebra $\FullEnvSymFer$ by the left ideal generated by $\AlgEta{k}$, $\AlgChi{k}$, $\AlgEtaBar{k}$ and $\AlgChiBar{k}$ with positive indices $k$, and the elements $\AlgEta{0} - \AlgEtaBar{0}$ and $\AlgChi{0} - \AlgChiBar{0}$.
The logarithmic Fock space is canonically a (left) \mbox{$\FullEnvSymFer$-module};
we refer as \term{current modes} to the generators~$\AlgChi{k}$, $\AlgEta{k}$, $\AlgChiBar{k}$ and $\AlgEtaBar{k}$ when they are regarded as linear operators on $\FullFock$.

The \term{ground states} are defined as in Section~\ref{subsec: ground states and currents}:
\begin{align*}
\GroundPartner \coloneqq \big[\, 1 \,\big],
\mspace{40mu}
\Ground \coloneqq \big[\, \AlgChi{0}\AlgEta{0} \,\big],
\mspace{40mu}
\GroundChi \coloneqq -\big[\, \AlgChi{0} \,\big],
\mspace{20mu}
\text{ and }
\mspace{25mu}
\GroundEta \coloneqq -\big[\, \AlgEta{0} \,\big].
\end{align*}
Similarly as in the chiral module, the currents are defined by
\begin{align*}
\mspace{20mu}
\HolCurrentChi
\,\coloneqq\,
\AlgChi{-1}\Ground\,,
\mspace{15mu}\phantom{\text{ and }}\mspace{20mu} & 
\HolCurrentEta
\,\coloneqq\,
\AlgEta{-1}\Ground\,,
\mspace{60mu}
& \phantom{\bigg\vert}(\text{\term{holomorphic currents}})
\\
\AntiHolCurrentChi
\,\coloneqq\,
\AlgChiBar{-1}\Ground\,,
\mspace{15mu}\text{ and }\mspace{20mu} & 
\AntiHolCurrentEta
\,\coloneqq\,
\AlgEtaBar{-1}\Ground\,.
\mspace{30mu}
& \phantom{\bigg\vert}(\text{\term{antiholomorphic currents}})
\end{align*}

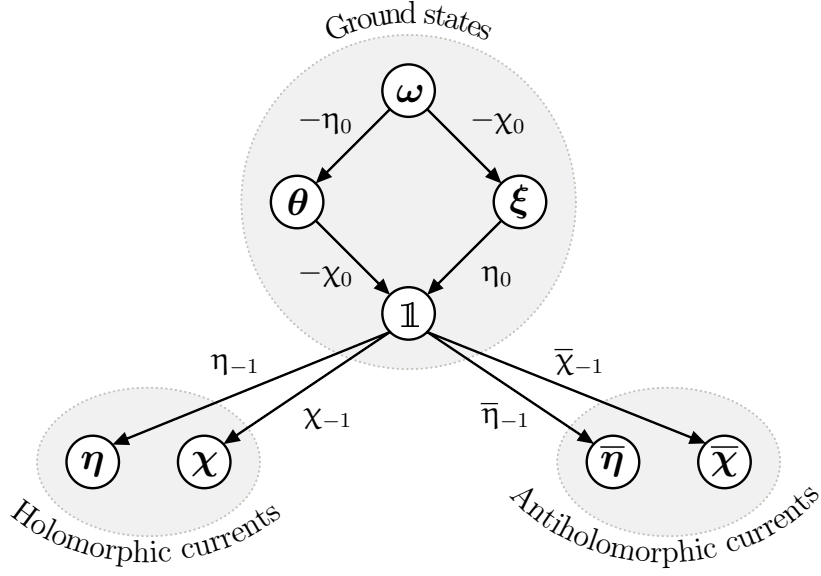
\begin{figure}[t!]
	\centering
	\vspace{10pt}
	\begin{overpic}[scale=0.778]{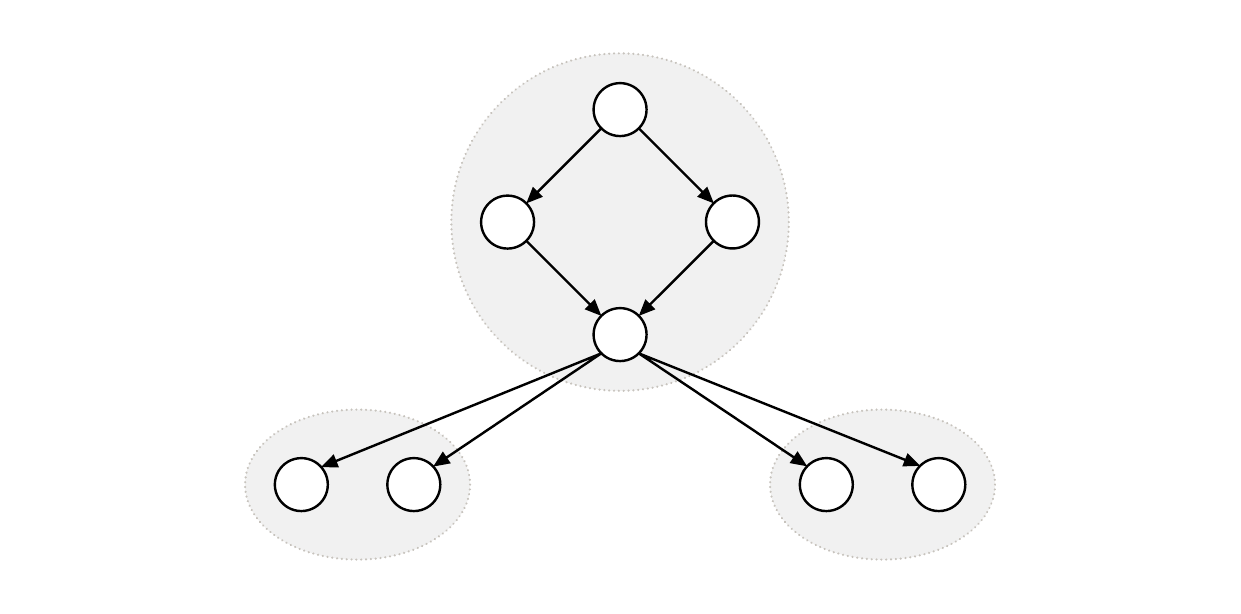}
		\put(49,39.7){\large $\GroundPartner$}
		\put(49.05,21.4){\large $\Ground$}
		\put(40.1,30.5){\large $\GroundEta$}
		\put(58.2,30.8){\large $\GroundChi$}
		\put(23.3,9.8){\large $\HolCurrentEta$}
		\put(32.3,9.8){\large $\HolCurrentChi$}
		\put(65.6,9.8){\large $\AntiHolCurrentEta$}
		\put(74.6,9.8){\large $\AntiHolCurrentChi$}
		\put(55,37.6){$-\AlgChi{0}$}
		\put(41,37.6){$-\AlgEta{0}$}
		\put(56,25.1){$\AlgEta{0}$}
		\put(41,25.1){$-\AlgChi{0}$}
		\put(34,18){$\AlgEta{-1}$}
		\put(41.5,13.6){$\AlgChi{-1}$}
		\put(56,13.6){$\AlgEtaBar{-1}$}
		\put(62,18){$\AlgChiBar{-1}$}
		\begin{tikzpicture}
		\filldraw[white] (0,0) circle (2pt);
		\node (A) at (6.05,5.5) [] {}; 
		\node (B) at (10.15,5.5) [] {};
		\node (C) at (2.3,0) [] {}; 
		\node (D) at (6.9,0) [] {};
		\node (E) at (9.2,0) [] {}; 
		\node (F) at (13.8,0) [] {};
		\def\myshift#1{\raisebox{-2.5ex}}
		\draw [->,thick,color=white,draw opacity = 0, postaction={decorate,decoration={text along path,text align=center,text={|\myshift|Ground states}}}] (A) to [bend left=45]  (B);
		\def\myshift#1{\raisebox{1ex}}
		\draw [->,thick,color=white,draw opacity = 0, postaction={decorate,decoration={text along path,text align=center,text={Holomorphic currents}}}]      (C) to [bend right=55] (D);
		\draw [->,thick,color=red,draw opacity = 0, postaction={decorate,decoration={text along path, text align=center,text={Antiholomorphic currents}}}]      (E) to [bend right=55] (F);
		\end{tikzpicture}
	\end{overpic}
	\centering
	\caption{
		The ground states, the currents, and \\ their relation through the action of the current modes.
	}
	\label{fig: ground states - full symalg}
\end{figure}

The logarithmic Fock space admits the PBW-type basis 
\begin{align}\label{eq: full basis}
\AlgEta{-k_r}\cdots\AlgEta{-k_2}\AlgEta{-k_1}
\AlgChi{-\ell_s}\cdots\AlgChi{-\ell_2}\AlgChi{-\ell_1}
\AlgEtaBar{-\bar{k}_{\bar{r}}}\cdots\AlgEtaBar{-\bar{k}_2}\AlgEtaBar{-\bar{k}_1}
\AlgChiBar{-\bar{\ell}_{\bar{s}}}\cdots\AlgChiBar{-\bar{\ell}_2}\AlgChiBar{-\bar{\ell}_1}
\GroundPartner
\end{align}
\vspace{-25pt}
\begin{align}
\nonumber
\text{ with } \qquad
& \phantom{\Big\vert}r,s,\bar{r},\bar{s}\in\Znneg \\
\nonumber
\text{ and the orderings } \qquad 
& \phantom{\Big\vert}{0 \leq k_1 < k_2 < \cdots < k_r}\,,\\
\nonumber
& \phantom{\Big\vert}{0 \leq \ell_1 < \ell_2 < \cdots < \ell_s}\,,\\
\nonumber
& \phantom{\Big\vert}{0 < \bar{k}_1 < \bar{k}_2 < \cdots < \bar{k}_{\bar{r}}}\,, \text{ and }\\
\nonumber
& \phantom{\Big\vert}{0 < \bar{\ell}_1 < \bar{\ell}_2 < \cdots < \bar{\ell}_{\bar{s}}}\,.
\end{align}

Similarly as in the chiral setting, the non-chiral logarithmic Fock space $\FullFock$ admits a natural $\Z_2$-grading in terms of the above basis that we call \term{parity}.
We set the parity to be $\boldsymbol{+}$~(plus) for the \term{bosonic states}
\begin{align}\label{eq: parity bos}
\BosFullFock
\coloneqq
\spn_\C
\Big\{\,
\AlgEta{-k_r}\cdots
\AlgChi{-\ell_s}\cdots
\AlgEtaBar{-\bar{k}_{\bar{r}}}\cdots
\AlgChiBar{-\bar{\ell}_{\bar{s}}}\cdots
\big[\, 1 \,\big]
\,\Big\vert\,
s + r + \bar{s} + \bar{r} = 0 \textnormal{ mod } 2
\,\Big\}\,,
\end{align}
and $\boldsymbol{-}$~(minus) for the \term{fermionic states}
\begin{align}\label{eq: parity fer}
\FerFullFock\,
\coloneqq
\spn_\C
\Big\{\,
\AlgEta{-k_r}\cdots
\AlgChi{-\ell_s}\cdots
\AlgEtaBar{-\bar{k}_{\bar{r}}}\cdots
\AlgChiBar{-\bar{\ell}_{\bar{s}}}\cdots
\big[\, 1 \,\big]
\,\Big\vert\,
s + r + \bar{s} + \bar{r} = 1 \textnormal{ mod } 2
\,\Big\}\,.
\end{align}

In Figure~\ref{fig: ground states - full symalg}, we show visually the ground states and the currents, and their relations via the action of the current modes.

The Sugawara construction is realized through the formulae
\begin{align*}
\SugL{n}
\coloneqq
\sum_{k \,\geq\, n/2}
\AlgChi{n-k}\AlgEta{k}
-
\sum_{k \,<\, n/2}
\AlgEta{k}\AlgChi{n-k}
\mspace{30mu}
\text{ and }
\mspace{30mu}
\SugLBar{n}
\coloneqq
\sum_{k \,\geq\, n/2}
\AlgChiBar{n-k}\AlgEtaBar{k}
-
\sum_{k \,<\, n/2}
\AlgEtaBar{k}\AlgChiBar{n-k}\,,
\end{align*}
for $n\in\Z$,
which give rise to well-defined linear operators \mbox{$\FullFock\rightarrow\FullFock$}.
The \term{holomorphic Virasoro modes}~$\SugL{n}$ and the \term{antiholomorphic Virasoro modes}~$\SugLBar{n}$ satisfy the Virasoro commutation relations with central charge~$-2$.
It is straightforward to see that the two copies of Virasoro modes commute with each other, \mbox{i.e.}, we have
\begin{align*}
\mathllap{\big[\, \SugL{n} , \SugLBar{m} \,\big]}
\; = \;
\mathrlap{0
	\mspace{90mu}
	\text{ for all } n,m\in\Z\,.}
\end{align*}

The nilpotent parts of the Virasoro modes~$\SugL{0}$ and $\SugLBar{0}$ are, respectively, $\SugL{0}^{(\text{nil})}=\AlgChi{0}\AlgEta{0}$ and $\SugLBar{0}^{(\text{nil})}=\AlgChiBar{0}\AlgEtaBar{0}$.
Then, the "extra" relations $\AlgChi{0}=\AlgChiBar{0}$ and $\AlgEta{0}=\AlgEtaBar{0}$ of the current modes ---that is, as linear operators on $\FullFock$--- in the quotient~\eqref{eq: full fock} conspicuously yield the Gaberdiel--Kausch condition~\mbox{$\SugL{0}^{(\text{nil})}-\SugLBar{0}^{(\text{nil})}=0$}.

\begin{figure}[h!]
	\centering
	\vspace{10pt}
	\begin{overpic}[scale=0.778]{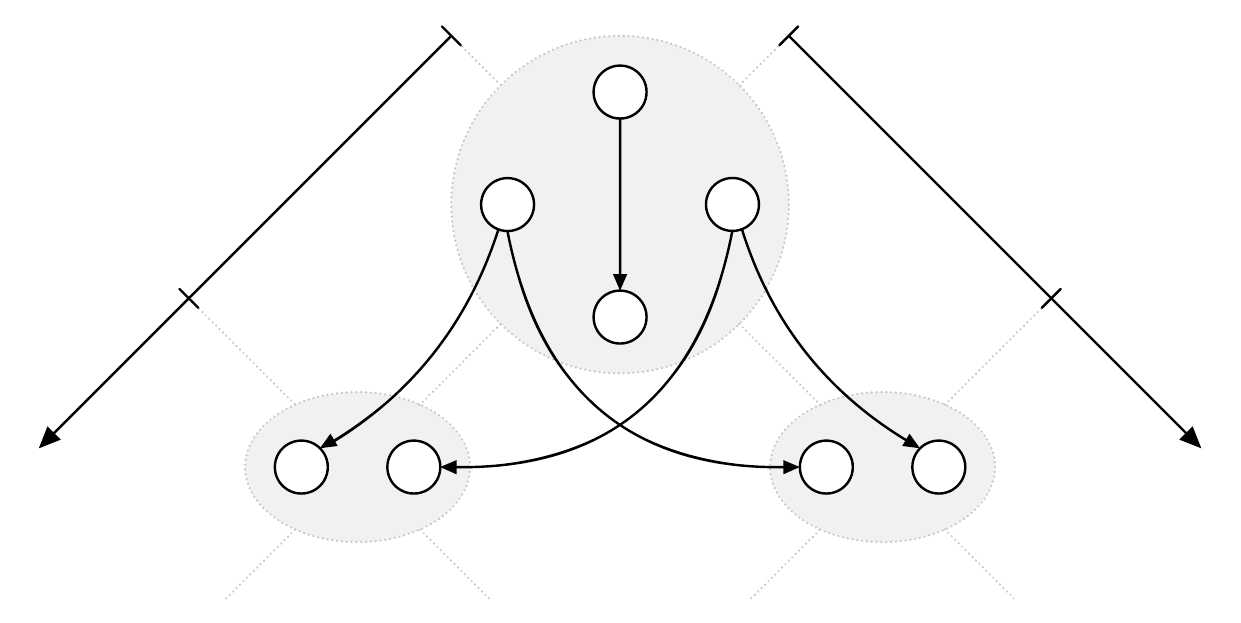}
		\put(49,41.85){\large $\GroundPartner$}
		\put(49.05,23.55){\large $\Ground$}
		\put(40.1,32.6){\large $\GroundEta$}
		\put(58.2,32.9){\large $\GroundChi$}
		\put(23.3,11.9){\large $\HolCurrentEta$}
		\put(32.3,11.9){\large $\HolCurrentChi$}
		\put(65.6,11.9){\large $\AntiHolCurrentEta$}
		\put(74.6,11.9){\large $\AntiHolCurrentChi$}
		\put(66,49.2){$0$}
		\put(33,49.2){$0$}
		\put(12,27.8){$1$}
		\put(87,27.8){$1$}
		\put(51,34.5){$\SugLBar{0}$}
		\put(46.8,34.5){$\SugL{0}$}
		\put(29.7,21.8){$\SugL{-1}$}
		\put(66.3,21.8){$\SugLBar{-1}$}
		\put(42,9.8){$\SugL{-1}$}
		\put(55,9.8){$\SugLBar{-1}$}
		\put(5,33.5){\rotatebox{45}{holomorphic}}
		\put(4,27.5){\rotatebox{45}{conformal dimension}}
		\put(82.5,45){\rotatebox{-45}{antiholomorphic}}
		\put(78.5,44){\rotatebox{-45}{conformal dimension}}
		\begin{tikzpicture}
		\filldraw[white] (0,-1) circle (2pt);
		\filldraw[white, opacity = 0.8] (3.65,-0.9) circle (5.9pt);
		\filldraw[white, opacity = 0.8] (5.63,-0.88) circle (5.7pt);
		\filldraw[white, opacity = 0.8] (10.57,-0.9) circle (5.8pt);
		\filldraw[white, opacity = 0.8] (10.48,-0.93) circle (5.8pt);
		\filldraw[white, opacity = 0.8] (12.6,-0.9) circle (5.6pt);
		\node (A) at (6.05,5.5) [] {}; 
		\node (B) at (10.15,5.5) [] {};
		\node (C) at (2.35,0) [] {}; 
		\node (D) at (6.95,0) [] {};
		\node (CC) at (2.35,-0.33) [] {}; 
		\node (DD) at (6.95,-0.33) [] {};
		\node (E) at (9.25,0) [] {}; 
		\node (F) at (13.85,0) [] {};
		\node (EE) at (9.25,-0.33) [] {}; 
		\node (FF) at (13.85,-0.33) [] {};
		\def\myshift#1{\raisebox{-2.5ex}}
		\draw [->,thick,color=white,draw opacity = 0, postaction={decorate,decoration={text along path,text align=center,text={|\myshift|Ground states}}}] (A) to [bend left=45]  (B);
		\def\myshift#1{\raisebox{1ex}}
		\draw [->,thick,color=white,draw opacity = 0, postaction={decorate,decoration={text along path,text align=center,text={Holomorphic}}}]      (C) to [bend right=55] (D);
		\draw [->,thick,color=white,draw opacity = 0, postaction={decorate,decoration={text along path,text align=center,text={currents.}}}]      (CC) to [bend right=55] (DD);
		\draw [->,thick,color=white,draw opacity = 0, postaction={decorate,decoration={text along path,text align=center,text={Antiholomorphic}}}]      (E) to [bend right=55] (F);
		\draw [->,thick,color=white,draw opacity = 0, postaction={decorate,decoration={text along path,text align=center,text={currents}}}]      (EE) to [bend right=55] (FF);
		\filldraw[white, opacity = 1] (5.37,-1.43) circle (1pt);
		\end{tikzpicture}
	\end{overpic}
	\centering
	\caption{
		The ground states, the currents, and \\ their relation through the action of the Virasoro modes.
	}
	\label{fig: ground states - full vir}
\end{figure}
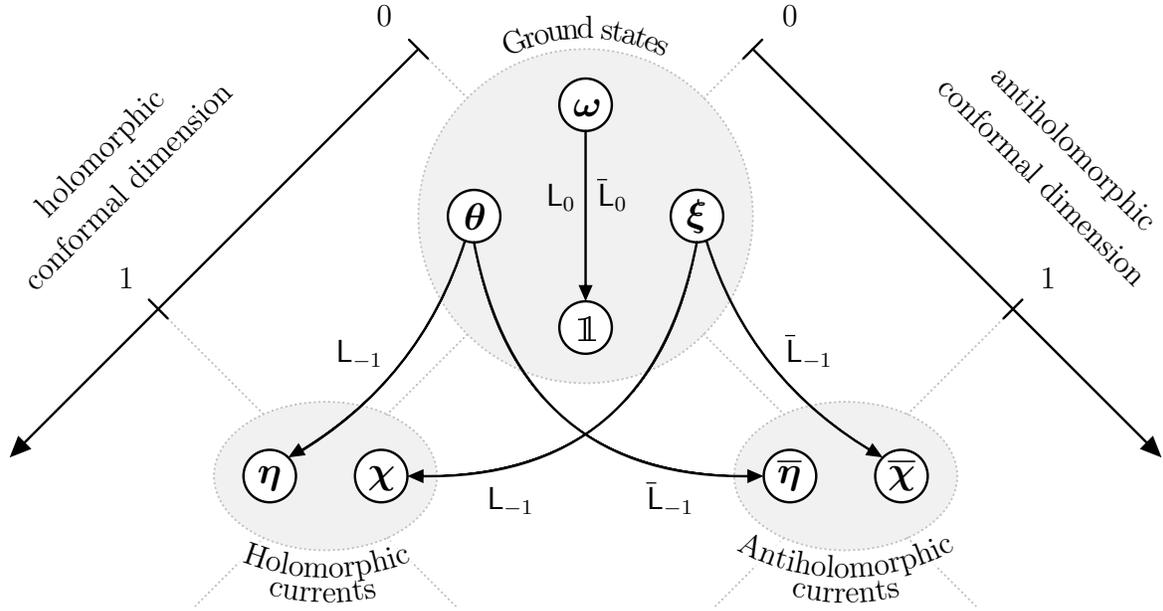

The operators~$\SugL{0}$ and $\SugLBar{0}$ provide a natural \mbox{$\Znneg$-bigrading} of $\FullFock$ via its decomposition into generalized eigenspaces.
The generalized $\SugL{0}$ and $\SugLBar{0}$-eigenvalues are called, respectively, \term{holomorphic} and \term{antiholomorphic conformal dimensions}.
The bigrading can be ex\-pres\-sed simply in terms of the basis~\eqref{eq: full basis}:
a straightforward computation leads to
\begin{align*}
\Big[\,
\SugL{0} -
\Delta_{\bk,\bl}\;\id_{\FullFock}
\Big]^2
\big(
\AlgEta{-k_r}\cdots\AlgEta{-k_1}
\AlgChi{-\ell_s}\cdots\AlgChi{-\ell_1}
\AlgEtaBar{-\bar{k}_{\bar{r}}}\cdots\AlgEtaBar{-\bar{k}_1}
\AlgChiBar{-\bar{\ell}_{\bar{s}}}\cdots\AlgChiBar{-\bar{\ell}_1}
\GroundPartner
\big)
\,=\, 0\,,
\end{align*}
and
\begin{align*}
\Big[\,
\SugLBar{0} -
\bar{\Delta}_{\bar{\bk},\bar{\bl}}\;\id_{\FullFock}
\Big]^2
\big(
\AlgEta{-k_r}\cdots\AlgEta{-k_1}
\AlgChi{-\ell_s}\cdots\AlgChi{-\ell_1}
\AlgEtaBar{-\bar{k}_{\bar{r}}}\cdots\AlgEtaBar{-\bar{k}_1}
\AlgChiBar{-\bar{\ell}_{\bar{s}}}\cdots\AlgChiBar{-\bar{\ell}_1}
\GroundPartner
\big)
\,=\, 0\,,
\end{align*}
where $\Delta_{\bk,\bl}\coloneqq\sum_{i=1}^rk_i+\sum_{j=1}^{s}\ell_j$
and
$\bar{\Delta}_{\bar{\bk},\bar{\bl}}\coloneqq\sum_{i=1}^{\bar{r}}\bar{k}_i+\sum_{j=1}^{\bar{s}}\bar{\ell}_j$ are the conformal dimensions of the basis state.
In particular, this implies that the Virasoro modes~$\SugL{0}$ and $\SugLBar{0}$ have Jordan blocks of rank at most~$2$.

Similarly as in the chiral setting, the currents are related to the fermions via the action of the Virasoro modes.
We record this fact in a remark for later reference.

\begin{rmk}\label{rmk: currents derivatives of fermions}
Similarly as in the chiral setting ---Example~\ref{ex: chiral currents}---, we have
	\begin{align*}
	\HolCurrentChi =  \SugL{-1} \GroundChi
	\,,
	\mspace{40mu}
	\HolCurrentEta = \SugL{-1} \GroundEta
	\,,
	\mspace{40mu}
	\AntiHolCurrentChi =  \SugLBar{-1} \GroundChi
	\,,
	\mspace{35mu}
	\text{ and }
	\mspace{35mu}
	\AntiHolCurrentEta =  \SugLBar{-1} \GroundEta\,.
	\end{align*}
	Within correlation functions, the above relations have the interpretation of the (anti)holo\-morphic currents being the (anti)holo\-morphic Wirtinger derivative of the corresponding fermion.
	This justifies the name \emph{currents} for these states.
	\hfill$\diamond$
\end{rmk}

In Figure~\ref{fig: ground states - full vir}, we show visually the ground states and the currents, and their relations via the action of the Virasoro modes.

\subsection{A glimpse of correlation functions}
\label{subsec: glimpse corr fun}

In Section~\ref{sec: corrfun} below, we will construct the correlation functions of the symplectic fermions CFT in domains of the complex plane.
In this section, we give some remarks that will guide us in such construction.

\begin{rmk}\label{rmk: automorphisms}
For each complex constant~$\thectt\in\C$, there exists an $\FullEnvSymFer$-module auto\-morphism, which
---by the $\FullEnvSymFer$-cyclicity of $\GroundPartner$ in $\FullFock$---
is fully determined by
\begin{align*}
	\GroundPartner \longmapsto \GroundPartner + \thectt\,\Ground\,,
\end{align*}
that is clearly not trivial for $\alpha \neq 0$.
These maps automatically lift to $(\Vir\oplus\VirBar)$-module automorphisms.
Note that they act as the identity map when restricted to the submo\-dules~$(\FullEnvSymFer)\GroundEta$, $(\FullEnvSymFer)\GroundChi$ and $(\FullEnvSymFer)\Ground$.
These automorphisms can be related to rescalings of the coordinates within correlation functions --- see Example~\ref{ex: rescaling} below.

\noindent
Physically, this means that a state cannot be distinguished from its image under one of these automorphisms.  
We should \emph{not} expected, then, the correlation functions of a CFT built from the representation $\FullFock$ to be canonically unique.
Rather, to determine the correlation functions of the theory we need to fix one extra parameter that can take as many values as there are non-equivalent automorphisms of $\FullFock$.
\hfill$\diamond$
\end{rmk}

Fix $\thectt\in\C$ and consider $n\in\Zpos$ fields ${\field_1,\ldots,\field_n\in\FullFock}$.
Their correlation function in a domain $\domain\subset\C$ is a $\C$-valued real-analytic multivariable function
\begin{align*}
(z_1,\ldots,z_n)
\longmapsto
\BigCorrFun{\domain}{\thectt}{\field_1(z_1)\cdots \field_n(z_n)}
\end{align*}
of distinct points $z_1,\ldots,z_n\in\domain$.
The collection of correlation functions of the theory is determined by the algebraic structure of the space~$\FullFock$.
In particular, a fundamental dictate of CFT says that the action of $\AlgL{-1}$ and $\AlgLBar{-1}$ translate, respectively, into holomorphic and antiholomorphic Wirtinger derivatives inside correlation functions. That is, we have
\begin{align*}
\BigCorrFun{\domain}{\thectt}{ \big[\SugL{-1}\field\big] (z)\cdots}
=
\partial_{z}\BigCorrFun{\domain}{\thectt}{ \field(z)\cdots}
\mspace{15mu}
\text{ and }
\mspace{25mu}
\BigCorrFun{\domain}{\thectt}{ \big[\mspace{1mu}\SugLBar{-1}\field\big] (z)\cdots}
=
\overline\partial_{z}\BigCorrFun{\domain}{\thectt}{ \field(z)\cdots}\,,
\end{align*}
for any field $\field\in\FullFock$.

Note, then, that the algebraic relations~$\SugL{-1}\field = 0$ and $\SugLBar{-1}\field = 0$ translate into the correlations functions of $\field\in\FullFock$ being antiholomorphic and holomorphic, respectively.

\begin{rmk}\label{rmk: identity field}
The ground state $\Ground$ satisfies both
\begin{align*}
	\SugL{-1} \Ground = 0
	\mspace{65mu}
	\text{ and }
	\mspace{65mu}
	\SugLBar{-1} \Ground = 0\,.
\end{align*} 
This implies that its correlation functions do not depend on the insertion point $z\in\domain$.
In Example~\ref{ex: identity}, we will determine the precise action of $\Ground$ within correlation functions.
\hfill$\diamond$
\end{rmk}

\begin{rmk}
	The currents $\HolCurrentChi$ and $\HolCurrentEta$ satisfy
	\begin{align*}
	\SugLBar{-1} \HolCurrentChi 
	\; = \;
	\SugLBar{-1} \HolCurrentEta 
	\; = \;
	0\,,
	\end{align*} 
	which justifies their name \emph{holomorphic} currents.
	Similarly, the \emph{antiholomorphic} currents $\AntiHolCurrentChi$ and $\AntiHolCurrentEta$ satisfy $\SugL{-1}\AntiHolCurrentChi = \SugL{-1}\AntiHolCurrentEta = 0$.
	\hfill$\diamond$
\end{rmk}

Lastly, the combined action of $\SugL{-1}$ and $\SugLBar{-1}$ can be interpreted in terms of the laplacian operator $\Delta\coloneqq \partial_x^2+\partial_y^2$, with $x = \Re(z)$ and $y = \Im(z)$, as a consequence of its factorization $\Delta=4\partial\overline\partial=4\overline\partial\partial$ in terms of the holomorphic and antiholomorphic Wirtinger derivatives.

\begin{rmk}\label{rmk: fermions harmonic}
	The ground fermions satisfy
	\begin{align*}
		\SugL{-1}\,\SugLBar{-1} \GroundChi \; = \; 0
		\mspace{70mu} 
		\text{ and }
		\mspace{70mu}
		\SugL{-1}\,\SugLBar{-1} \GroundEta \; = \; 0\,.
	\end{align*}
	These algebraic identities translate to the correlation functions of the ground fermions being harmonic with respect to its insertion points.
	\hfill$\diamond$
\end{rmk}

\newpage
\section{Correlation functions}%
\label{sec: corrfun}
In this section, we define the correlation functions for the symplectic fermions CFT in domains of the complex plane.
In particular, we construct correlation functions on (non-empty) proper simply-connected open subsets of the complex plane.
For the remainder of the text, any domain~$\domain\subset\C$ is taken to be of such kind.

Correlation functions are the physically relevant quantities in a field theory.
In the context of statistical mechanics, one can think of CFT correlation functions as the scaling limit of the ---appropriately renormalised--- correlation functions of a lattice model.
However, this picture is not so simple in logarithmic CFT.
The presence of logarithmic fields ---whose correlation functions feature logarithmic dependencies--- render the correlation functions of the theory not exactly physical.
The origin of this can be seen more transparently in algebraic terms:
there exists a family of non-trivial Virasoro automorphisms of the space of fields.

In Remark~\ref{rmk: automorphisms}, we encountered a family of non-trivial automorphisms of $\FullFock$ indexed by a constant $\thectt\in\C$.
Essentially, this means that we cannot physically distinguish a field and its image under these automorphisms.
We also noted in Remark~\ref{rmk: automorphisms}, that these automorphisms act as the identity on the submodules~$(\FullEnvSymFer)\GroundEta$, $(\FullEnvSymFer)\GroundChi$ and $(\FullEnvSymFer)\Ground$.
We therefore expect the correlation functions of the fields in these submodules to be canonically defined.
However, to completely determine the correlation functions of the logarithmic fields ---that is, those that are not invariant under the automorphisms of $\FullFock$--- we need to specify an extra piece of data; in particular, we need to specify a constant $\thectt\in\FullFock$.

Fixing $\thectt\in\C$, the \term{correlation functions} of the symplectic fermions are a collection, indexed by $n\in\Zpos$, of linear maps
\begin{align*}
	\bigCorrFun{\domain}{\thectt}{\cdots}
	\; \colon \;
	\FullFock^{\,\otimes n} \longrightarrow \AnalyFun{\Conf{n}{\domain}}{\C}\,,
\end{align*}
where $\AnalyFun{\Conf{n}{\domain}}{\C}$ is the set of $\C$-valued real-analytic functions on the \term{configuration space}
\begin{align*}
	\Conf{n}{\domain}
	\coloneqq
	\Big\{
	(z_1,\ldots,z_n)\in\domain^n \,\big\vert\, z_i \neq z_j \text{ for } i\neq j
	\Big\}\,
\end{align*}
of the domain~$\domain$.
In other words, for $n\geq1$ fixed fields~$\field_1,\ldots,\field_n\in\FullFock$, we have a function
\begin{align*}
	(z_1,\ldots,z_n)
	\longmapsto
	\BigCorrFun{\domain}{\thectt}{\field_1(z_1)\cdots\field_n(z_n)}\,
\end{align*}
which is real-analytic on the variables $z_i\in\domain$ away from one another.

In Theorem~\ref{thm: characterization of CF}, we provide the characterization of the correlation functions of the symplectic fermions.
To do so, we follow two basic principles of CFT.
The first one, explored more in detail in Section~\ref{subsec: L-1} below, states what the action of the Virasoro modes~$\SugL{-1}$ and $\SugLBar{-1}$ translates into within correlation functions.
The second one involves the existence of expansions of the type
\begin{align*}
	\BigCorrFun{\domain}{\thectt}{\field_1(z_1)\field_2(z_2)\cdots}
	\;=\;
	\sum_{i} f_{i}(z_1-z_2)\BigCorrFun{\domain}{\thectt}{\field_i(z_2)\cdots}
\end{align*}
that are convergent for $\vert z_1 - z_2\vert$ small enough.
These are known as operator product ex\-pan\-sions (OPEs); we describe them in more detail in Section~\ref{subsec: OPE}.
\black

\subsection{Conformal covariance}
\label{subsec: L-1}

One of the axioms of CFT that connects the conformal symmetries of correlation functions with the algebraic structure of the fields is the following:
for any field $\field\in\FullFock$, we have
\begin{align}
\mspace{20mu}
&
\phantom{\Bigg\vert}
\BigCorrFun{\domain}{\thectt}{ \big[\SugL{-1}\field\big] (z)\cdots}
\;=\;
\partial_{z}\BigCorrFun{\domain}{\thectt}{ \field(z)\cdots}\,,
\mspace{5mu}
\text{ and }
\label{eq: L-1}
\\
\label{eq: L-1 bar}
&
\phantom{\bigg\vert}
\BigCorrFun{\domain}{\thectt}{ \big[\mspace{1mu}\SugLBar{-1}\field\big] (z)\cdots}
\;=\;
\overline\partial_{z}\BigCorrFun{\domain}{\thectt}{ \field(z)\cdots}\,,
\end{align}
where $\partial$ and $\overline\partial$ are the holomorphic and antiholomorphic Wirtinger derivatives, respectively.

In Remark~\ref{rmk: fermions harmonic}, we anticipated that the identities
\begin{align}\label{eq: fermions harmonic alg}
\SugL{-1}\,\SugLBar{-1} \GroundChi \; = \; 0
\mspace{70mu} 
\text{ and }
\mspace{70mu}
\SugL{-1}\,\SugLBar{-1} \GroundEta \; = \; 0\,,
\end{align}
would guide our construction of the correlation functions of the theory.
Consider the $2$-point function
\begin{align}\label{eq: 2pt funct}
	\BigCorrFun{\domain}{\thectt}{\GroundChi(z)\,\GroundEta(w)}
\end{align}
of the ground fermions, and recall that the laplacian operator~$\Delta\coloneqq \partial_x^2+\partial_y^2$ can be factorized as ${\Delta=4\partial\overline\partial=4\overline\partial\partial}$.
Then, from~\eqref{eq: L-1} and \eqref{eq: L-1 bar}, we have that the identities~\eqref{eq: fermions harmonic alg} imply that the function
\begin{align*}
(z,w)
\longmapsto
\BigCorrFun{\domain}{\thectt}{\GroundChi(z)\,\GroundEta(w)}\,,
\end{align*}
is harmonic in each of its variables separately in $\Conf{2}{\domain}$.
Its singular behaviour when $\vert z - w \vert \to 0$ is determined to be proportional to $\log \vert z - w \vert$ by the conformal covariance of the theory (\cite{Kausch-SF}).
The exact value of the two-point function, however, depends on the boundary values we choose on the domain~$\domain$.
Consistently with this, in Theorem~\ref{thm: characterization of CF}, we fix it to be proportional to the Green's function of the laplacian with Dirichlet boundary conditions.

For concreteness, the laplacian Green's function with Dirichlet boundary conditions on a domain~$\domain$ is the real-analytic function~$\Green_\domain$ on $\Conf{2}{\domain}$ that satisfies
\begin{align}\label{eq: Green functions}
& \mspace{8mu}
\Green_\domain(z,w) = \frac{1}{2\pi}\log\frac{1}{\vert z-w\vert} + \HarmOfGreen_\domain(z,w)\,,
\mspace{30mu}
\text{ and }
\\
\nonumber
& \lim_{z\to\bdry\domain} \Green_\domain(z,w)  = 0\,, \phantom{\Bigg\vert}
\end{align}
with $\HarmOfGreen_\domain \colon \domain\times\domain \longrightarrow \R$ a harmonic function on both variables separately.
In particular, the function $ z \mapsto \HarmOfGreen_{\domain} (z,w)$ on $\domain$ is the (unique) harmonic extention of the boundary values $x \mapsto \frac{1}{2\pi}\log \vert x - w \vert$.
Note too that $\Green_\domain$ is harmonic in both variables separately on its domain of definition~$\Conf{2}{\domain}$.

Besides the harmonicity of the ground fermions, we also anticipated in Remark~\ref{rmk: identity field} the identities~$\SugL{-1}\Ground = 0$ and $\SugLBar{-1}\Ground = 0$.
Under the light of \eqref{eq: L-1} and \eqref{eq: L-1 bar}, they imply that the correlation function
\begin{align*}
	\BigCorrFun{\domain}{\thectt}{\Ground(z)\cdots}
\end{align*}
does not depend on the insertion point~$z\in\domain$.
In Theorem~\ref{thm: characterization of CF} below, we will see that the field~$\Ground$ acts as the identity field within correlation functions, i.e. $\bigCorrFun{\domain}{\thectt}{\Ground(z)\cdots} = \bigCorrFun{\domain}{\thectt}{\cdots}$.

\subsection{Operator product expansions}
\label{subsec: OPE}
In the bootstrap approach to field \mbox{theory}, one circumvents the lagrangian formalism and computes the correlation functions of the theory by exploiting certain properties thereof.
A celebrated feature of conformal field theories that make it possible to bootstrap them is the existence of an \term{operator product expansion} (\term{OPE}); which in the present case is defined as follows.
For any two fields~${\field_1,\field_2\in\FullFock}$, there exists a countable collection~$\{f_i,\field_i\}_{i\in I}$ of $\C$-valued functions~$f_i$ on $\C\setminus\{0\}$, and fields $\field_i\in\FullFock$, such that
\begin{align}
\nonumber
\text{for } \qquad
& \text{all domains~$\domain \subset \C$}\,, \text{ and }\phantom{\Big\vert} \\
\nonumber
& \text{any  fields~} \field_3 , \ldots, \field_n \in \FullFock\,,
\phantom{\Big\vert}  \\
\nonumber
\text{we have } \qquad 
& \BigCorrFun{\domain}{\thectt}{\field_1(z_1)\field_2(z_2)\cdots\field_n(z_n)}
\phantom{\Bigg\vert}=\phantom{\Bigg\vert}
\sum_{i\in I} f_i(z_1-z_2)\BigCorrFun{\domain}{\thectt}{\field_i(z_2)\cdots\field_n(z_n)}
\\
\nonumber
\text{ whenever } \qquad
&  0 < \vert z_1 - z_2 \vert < \min \big\{ \!\textstyle\min_{j\geq 3} \vert z_2 - z_j\vert\,,\,\dist(z_2,\bdry\domain)\big\}\,,\phantom{\bigg\vert}
\end{align}
where $\dist(z_2,\bdry\domain)$ denotes the distance from the point $z_2$ to the boundary of $\domain$.

In such case, we then just write
\begin{align*}
\field_1(z_1)\field_2(z_2)
\,=\,
\;\sum_{i\in I} f_i(z_1-z_2) \, \field_i(z_2)\,.
\end{align*}

Naturally, since the correlation functions of the symplectic fermions are canonically defined only up to the field transformations in Remark~\ref{rmk: automorphisms},
only fields modulo automorphism appear in the OPEs of the theory.
See the statement of Theorem~\ref{thm: characterization of CF} below for an example of this fact.

In the previous section, we stressed that, in algebraic terms, 
the space of fields~$\FullFock$ is said to be logarithmic because
the Virasoro modes~$\SugL{0}$ and $\SugLBar{0}$ are non-diagonalizable thereon.
In terms of OPEs, the logarithmic nature of the theory can be seen more transparently.
Some of the expansions feature logarithms
---see Theorem~\ref{thm: characterization of CF}---,
in contrast with non-logarithmic CFTs, in which the OPE coefficients are of the type $(z_1-z_2)^{\alpha} {(\overline{z}_1-\overline{z}_2)^{\beta}}$ for $\alpha,\beta \in \R$.

Commonly in axiomatic descriptions of CFT, it is taken as an axiom that there exists a pair of fields ---called the holomorphic and antiholomorphic stress--energy tensors--- whose OPEs with other fields involve the action of the Virasoro modes.
However, in the CFT at hand, the Virasoro structure is a consequence of the symplectic fermions structure.
We, hence, formulate this second axiom in terms of the symmetry algebra.
In particular, in our characterization of the correlation functions
---Theorem~\ref{thm: characterization of CF}---,
it will be enough to require that the currents have the OPEs
\begin{align*}
	\HolCurrentChi(z)\field(w)
	\,=\,
	\sum_{k\in\Z}\,
	\frac{\,\big[\AlgChi{k}\field\big](w)\,}{\,(z-w)^{k+1}}\,,
	\mspace{130mu}
	\HolCurrentEta(z)\field(w)
	\,=\,
	\sum_{k\in\Z}\,
	\frac{\,\big[\AlgEta{k}\field\big](w)\,}{\,(z-w)^{k+1}}\,,
	\\
	\phantom{.}
	\\
	\AntiHolCurrentChi(z)\field(w)
	\,=\,
	\sum_{k\in\Z}\,
	\frac{\,\big[\AlgChiBar{k}\field\big](w)\,}{\,(\overline{z}-\overline{w})^{k+1}}\,,
	\mspace{40mu}
	\text{ and }
	\mspace{50mu}
	\AntiHolCurrentEta(z)\field(w)
	\,=\,
	\sum_{k\in\Z}\,
	\frac{\,\big[\AlgEtaBar{k}\field\big](w)\,}{\,(\overline{z}-\overline{w})^{k+1}}\,,
\end{align*}
with any other field~$\field\in\FullFock$.

It is sometimes more useful not to state all the terms in an OPE but rather the non-vanishing terms as $z_1\to z_2$.
For this reason, we introduce the notation 
\begin{align*}
\phantom{\bigg\vert}
\field_1(z_1)\field_2(z_2)
\,\OPE\,
\;\sum_{i\in I} f_i^{(\thectt)}(z_1-z_2) \field_i(z_2)
\end{align*}
to indicate
\begin{align*}
\phantom{\Bigg\vert}
\BigCorrFun{\domain}{\thectt}{\field_1(z_1)\field_2(z_2)\cdots\field_n(z_n)}
\,=\,
\sum_{i\in I}
f_i^{(\thectt)}(z_1-z_2)
\BigCorrFun{\domain}{\thectt}{\field_i(z_2)\cdots\field_n(z_n)}
\;+\;
o\big(\vert z_1 - z_2 \vert\big)
\end{align*}
as $\vert z_1 - z_2\vert \to 0$, for all fields~$\field_3,\ldots,\field_n\in\FullFock$ and any domain~$\domain$.	

\subsection{Characterization}
\label{subsec: characterization}

In the following theorem, we provide a characterization of the correlation functions of the symplectic fermions CFT.

\setlist[description]{leftmargin=2.6cm,labelindent=0.1cm,rightmargin=0cm}
\begin{thm}\label{thm: characterization of CF}
	For each complex number $\thectt\in\C$, there exists a unique collection
	\begin{align*}
		\bigg\{\;
		\bigCorrFun{\domain}{\thectt}{\cdots} 
		\  \colon \;
		\FullFock^{\,\otimes n} \longrightarrow \AnalyFun{\Conf{n}{\domain}}{\C}
		\;\bigg\}_{n\in\Zpos}
	\end{align*}
	of linear maps that satisfy the following properties:
\begin{description}[style=sameline]
			
	\item[(FER)]
	The correlation functions of the ground fermions are
	\begin{align*}
		\BigCorrFun{\domain}{\thectt}{\GroundChi(z_1)\,\GroundEta(w_1)\cdots\GroundChi(z_n)\,\GroundEta(w_n)}
		\, = \,
		(4\pi)^n
		\displaystyle\sum_{\sigma\in\mathcal{S}_n} \sgntr(\sigma)
		\displaystyle\prod_{i=1}^n
		\Green_\domain\big(z_i,w_{\sigma(i)}\big)
	\end{align*}
	for $(z_1,\ldots,z_n,w_1\ldots,w_n)\in\Conf{2n}{\domain}$,
	they vanish when the number of
	\linebreak[4]
	$\GroundChi$-inser\-tions and $\GroundEta$-insertions do not coincide,
	and they satisfy
	\begin{align*}
		\BigCorrFun{\domain}{\thectt}{\cdots\GroundChi(z)\,\GroundEta(w)\cdots}
		\,=
		-\,\BigCorrFun{\domain}{\thectt}{\cdots\GroundEta(w)\,\GroundChi(z)\cdots}\,.
	\end{align*}

	\item[(DER)]
	For the Virasoro modes~$\SugL{-1}$ and $\SugLBar{-1}$ we have
	\begin{align*}
		\BigCorrFun{\domain}{\thectt}{\big[\SugL{-1}\field\big](z)\cdots}
		\  = \ 
		\partial_z\,\BigCorrFun{\domain}{\thectt}{\field(z)\cdots}	
	\end{align*}
	and
	\begin{align*}
		\;\BigCorrFun{\domain}{\thectt}{\big[\mspace{1mu}\SugLBar{-1}\field\big](z)\cdots}
		\  = \ 
		\overline{\partial}_z\,\BigCorrFun{\domain}{\thectt}{\field(z)\cdots}\,,
	\end{align*}
	for any field~$\field\in\FullFock$.

\item[(CUR)] 
	The currents have the OPEs
	\begin{align*}
		\HolCurrentChi(z)\field(w)
		\,=\,
		\sum_{k\in\Z}\,
		\frac{\,\big[\AlgChi{k}\field\big](w)\,}{\,(z-w)^{k+1}}\,,
		\mspace{20mu}
		&
		\mspace{59mu}
		\AntiHolCurrentChi(z)\field(w)
		\,=\,
		\sum_{k\in\Z}\,
		\frac{\,\big[\AlgChiBar{k}\field\big](w)\,}{\,(\overline{z}-\overline{w})^{k+1}}\,,
		\\
		\phantom{.}
		\\
		\HolCurrentEta(z)\field(w)
		\,=\,
		\sum_{k\in\Z}\,
		\frac{\,\big[\AlgEta{k}\field\big](w)\,}{\,(z-w)^{k+1}}\,,
		\mspace{20mu}
		&
		\text{ and }
		\mspace{20mu}
		\AntiHolCurrentEta(z)\field(w)
		\,=\,
		\sum_{k\in\Z}\,
		\frac{\,\big[\AlgEtaBar{k}\field\big](w)\,}{\,(\overline{z}-\overline{w})^{k+1}}
	\end{align*}
for any field $\field\in\FullFock$.
	\item[(LOG)] 
	The ground fermions have the OPE
	\begin{align*}
		\GroundChi(z)\GroundEta(w)
		\,\OPE\,
		\log \frac{1}{\vert z-w\vert^2} \,\Ground(w)
		\,-\,
		\big[\GroundPartner + \thectt \, \Ground\big](w)
		\,.
	\end{align*}
	\end{description}
\end{thm}

The proof ---provided below in Section~\ref{subsec: proof}--- goes, in broad strokes, as follows.
Using the properties~\textbf{(FER)} and \textbf{(LOG)} one can obtain, recursively, the correlation functions of an arbitrary number of insertions of~$\GroundChi$, $\GroundEta$ and $\GroundPartner$.
From those correlations, and since the state~$\GroundPartner$ is $\FullEnvSymFer$-cyclic, we can use the properties~\textbf{(DER)} and \textbf{(CUR)} to get the correlation functions of all other fields.

Before providing the details of the proof, we make some remarks and give some examples regarding correlation functions---they can be easily extracted from the proof in Section~\ref{subsec: proof}.

Let $\parity\colon\FullFock\rightarrow\{\boldsymbol{+},\boldsymbol{-}\}$ be the parity defined in~\eqref{eq: parity bos} and \eqref{eq: parity fer}.
That is, the map~$\parity$ has the value $\boldsymbol{+}$~(plus) on $\BosFullFock$ and the value $\boldsymbol{-}$~(minus) on $\FerFullFock$.
The correlation functions of parity-homogeneous fields satisfy the following property.

\begin{rmk}
	Let $\field_1,\field_2\in\FullFock$ be two parity-homogeneous fields. 
	For $(z_1,z_2)\in\Conf{2}{\domain}$, we have
	\begin{align*}
	\BigCorrFun{\domain}{\thectt}{\cdots \field_1(z_1) \field_2(z_2) \cdots}
	\ = \,
	\susyfactor_{\parity(\field_1) , \parity(\field_2)}\;
	\BigCorrFun{\domain}{\thectt}{\cdots \field_2(z_2) \field_1(z_1) \cdots}\,,
	\end{align*}
	where $\susyfactor_{\boldsymbol{+},\boldsymbol{+}}=\susyfactor_{\boldsymbol{+},\boldsymbol{-}}=\susyfactor_{\boldsymbol{-},\boldsymbol{+}}=1$ and $\susyfactor_{\boldsymbol{-},\boldsymbol{-}}=-1$.
	In other words, within correlation functions, fermionic fields anticommute with each other, and bosonic fields commute with any other field.
	\hfill
	$\diamond$
\end{rmk}

We also give some examples of correlation functions.

\begin{ex}\label{ex: identity}
	The ground field~$\Ground$ acts as the identity field within correlation functions.
	That is,
	for any fields~$\field_1,\ldots,\field_n\in\FullFock$ and any insertion points~$(z_1,\ldots,z_n)\in\Conf{n}{\domain}$ distinct from $z\in\domain$, we have
	\begin{align*}
		\BigCorrFun{\domain}{\thectt}{\Ground(z) \field_1(z_1)
			\cdots \field_n(z_n)}
		=\,
		\BigCorrFun{\domain}{\thectt}{\field_1(z_1)
			\cdots \field_n(z_n)}\,.
	\end{align*}
	We also have%
	\footnote{
		One can find instances in the literature where the one-point function of $\Ground$ vanishes.
		Our choice is consistent with conceiveing the symplectic fermions (full) CFT as the scaling limit of a model of statistical mechanics---see \cite{AdRu-fDGFF_to_symfer}.}
	$\bigCorrFun{\domain}{\thectt}{\Ground(z)}=1$ for any $z\in\domain$
	This justifies the name of the field $\Ground$ to be the \emph{identity field}.
	\hfill
	$\diamond$
\end{ex}

\begin{ex}\label{ex: partner 1 pt}
	For any domain~$\domain$ and insertion point~$z\in\domain$, we have
	\begin{align*}
	\BigCorrFun{\domain}{\thectt}{\GroundPartner(z)}
	=
	- \big(\,4\pi\HarmOfGreen_{\domain}(z,z) + \thectt \,\big),
	\end{align*}
	where $\HarmOfGreen_\domain$ is the regular term of the Green's function~$\Green_\domain$ as defined in~\eqref{eq: Green functions}.
	Note that this one-point function is related to the conformal radius~$\mathrm{R}(z;\domain)\coloneqq \exp\big(2\pi\HarmOfGreen_\domain(z,z)\big)$ of the domain~$\domain$ from the point~$z\in\domain$.
	\hfill
	$\diamond$
\end{ex}

The covariance under conformal transformations of the field~$\GroundPartner$ is of logarithmic type.
Let us exemplify this with its one-point function.

\begin{ex}\label{ex: partner 1 pt trans}
Let $\phi \colon \domain \rightarrow \phi(\domain)$ be a conformal map between the domains~$\domain$ and $\phi(\domain)$.
We have
\begin{align*}
	\BigCorrFun{\phi(\domain)}{\thectt}{\GroundPartner\big(\phi(z)\big)}
	=\;
	\BigCorrFun{\domain}{\thectt}{\GroundPartner(z)}
	-\;
	\frac{1}{2\pi}\log\big\vert\phi'(z)\big\vert
\end{align*}
for any insertion point $z\in\domain$.
This follows from the conformal invariance of the Green's function, that is,
\begin{align*}
	\Green_\domain ( z , w )
	\,=\,
	\Green_{ \phi (\domain) } \big( \phi(z) , \phi(w) \big)\,,
\end{align*}
and the explicit value of the of the one-point function of $\GroundPartner$ ---Example~\ref{ex: partner 1 pt}---.
\hfill
$\diamond$
\end{ex}

The automorphisms of $\FullFock$ in Remark~\ref{rmk: automorphisms} can be related to scaling transformations within correlation functions.
Let us exemplify this with the one-point function of the ground field~$\GroundPartner$
under a rescaling of the upper-half plane~$\UpperHalf \coloneqq \{\im z > 0\}$.

\begin{ex}\label{ex: rescaling}
Fix $\lambda > 0$, and consider the conformal map~${\UpperHalf \to \UpperHalf}$ given by the rescaling $z \mapsto \lambda z$.
From Examples~\ref{ex: partner 1 pt} and \ref{ex: partner 1 pt trans},
we get
\begin{align*}
	\BigCorrFun{\UpperHalf}{\thectt}{\GroundPartner(\lambda z)}
	& =\;
	\BigCorrFun{\UpperHalf}{\thectt}{\GroundPartner(z)}
	-\phantom{\bigg\vert}
	\frac{1}{2\pi}\log\lambda\,
	\BigCorrFun{\UpperHalf}{\thectt}{\Ground(z)}
	\\
	& =\;
	\BigCorrFun{\UpperHalf}{\thectt}{\big[\GroundPartner-\thectt_\lambda\Ground\big](z)}
	\phantom{\Bigg\vert}
	\\
	& =\;
	\BigCorrFun{\UpperHalf}{(\thectt+\thectt_\lambda)}{\GroundPartner(z)}
	\phantom{\bigg\vert}
\end{align*}
with $\exp(2\pi\thectt_\lambda) = \lambda$.
\hfill$\diamond$
\end{ex}

\subsection{Proof of Theorem~\ref{thm: characterization of CF}}
\label{subsec: proof}

We provide here the details of the postponed proof.

\begin{proof}[Proof of Theorem~\ref{thm: characterization of CF}]
We break down the proof into three steps.
In Steps~1 and~2, we construct the correlation functions of fields in the basis~\eqref{eq: full basis} in a way that is uniquely determined by the properties~\textbf{(FER)}, \textbf{(DER)},
\textbf{(CUR)}, and~\textbf{(LOG)}.
Then, in Step~3, we check that the correlation functions obtained in Step~2 indeed satisfy the listed properties.

\noindent \textit{Step 1}: Let $\mathbf{z}\coloneqq\{z_1,\ldots,z_n\}$, $\mathbf{w}\coloneqq\{w_1,\ldots,w_n\}$ and $\mathbf{x}\coloneqq\{x_1,\ldots,x_k\}$ be three non-intersecting collections of distinct points in $\domain$,
and let ${\Bij(\mathbf{z},\mathbf{w};\mathbf{x})}$ be the set of bijections from the set~${\mathbf{z}\cup\mathbf{x}}$ onto the set~${\mathbf{w}\cup\mathbf{x}}$.
Our first step consists on proving the formula
\begin{align}\label{eq: formula step 1}
	\BigCorrFun{\domain}{\thectt}
	{\GroundChi(z_1)\GroundEta(w_1)&\cdots\GroundChi(z_n)\GroundEta(w_n)
	\GroundPartner(x_k)\cdots\GroundPartner(x_1)}
	\\\nonumber
	& =\phantom{\bigg\vert}
	(-1)^k\sum_{b\in\Bij(\mathbf{z},\mathbf{w};\mathbf{x})}
	(-1)^b(-1)^{F_b}
	\prod_{\underset{b(y)\neq y}{y\in\mathbf{z}\cup\mathbf{x}}}
	\BigCorrFun{\domain}{\thectt}
	{\GroundChi(y)\GroundEta\big(b(y)\big)}
	\prod_{\underset{b(x)= x}{x\in\mathbf{x}}}
	\BigCorrFun{\domain}{\thectt}
	{\GroundPartner(x)}\,,
\end{align}
where $(-1)^b$ is the signature of the bijection~$b$, and $F_b$ is the number of fixed points of $b$ ---that is, the number of factors in the last product---.
From the arguments that we provide below, we get two important byproducts.  Firstly, the order of the fields in~\eqref{eq: formula step 1} is not relevant so long as the relative order of the $\GroundChi$-insertions and $\GroundEta$-insertions is not altered.
Secondly, the correlation function~\eqref{eq: formula step 1} vanishes when the number of $\GroundChi$-insertions and $\GroundEta$-insertions is not the same.
We proceed recursively by induction on the number~$k$ of $\GroundPartner$-insertions; the formula is trivially checked to hold in the base case $k=0$ when we recover the expression in property~\textbf{(FER)}.
We can use property~\textbf{(LOG)} to find the correlation function with one more $\GroundPartner$-insertion, but, in order to do so, we need to insert the identity field~$\Ground$ too.
Recalling the relation~$\Ground=-\AlgChi{0}\GroundEta$, we can use the properties \textbf{(CUR)} and \mbox{\textbf{(DER)}}, to get
\begin{align*}
\BigCorrFun{\domain}{\thectt}
{\cdots
	\GroundChi(z_n)\GroundEta(w_n)
	\Ground(x_{k+1})\GroundPartner(x_k)
	\cdots}
& \,=-\,
\BigCorrFun{\domain}{\thectt}
{\cdots
	\GroundChi(z_n)\GroundEta(w_n)
	\big[\,\AlgChi{0}\GroundEta\,\big](x_{k+1})\GroundPartner(x_k)
	\cdots}
\\
& \,=-\,
\oint_{\gamma}
\frac{\dd\zeta}{2\pi\ii}
\BigCorrFun{\domain}{\thectt}{\cdots\HolCurrentChi(\zeta)\GroundEta(x_{k+1})
	\cdots}
\phantom{\bigg\vert}
\\
& \,=-\,
\oint_{\gamma}
\frac{\dd\zeta}{2\pi\ii}\,
{\partial_\zeta}
\BigCorrFun{\domain}{\thectt}{\cdots\GroundChi(\zeta)\GroundEta(x_{k+1})\cdots}
\,,
\phantom{\Bigg\vert}
\end{align*}
where $\gamma$ is a positively-oriented contour around $x_{k+1}\in\domain$ small enough not to encircle any other insertion points in the correlation function at hand.
In this case, we can rewrite formula~\eqref{eq: formula step 1} as
\begin{align*}
\BigCorrFun{\domain}{\thectt}{\cdots\GroundChi(&z_n)\GroundEta(w_n)\GroundChi(\zeta)\GroundEta(x_{k+1})\GroundPartner(x_k)\cdots}
=
\\
&
\BigCorrFun{\domain}{\thectt}{\GroundChi(\zeta)\GroundEta(x_{k+1})}
\BigCorrFun{\domain}{\thectt}{\cdots\GroundChi(z_n)\GroundEta(w_n)\GroundPartner(x_k)\cdots}
+
\text{ harmonic w.r.t $\zeta$ at $x_{k+1}$}\,.\phantom{\bigg\vert}
\end{align*}
Recalling the properties of the Green's function~\eqref{eq: Green functions} and the factorization of the laplacian~$\Delta=4\partial \overline\partial$, the integral can be evaluated to get
\begin{align}\label{eq: indentity field}
	\BigCorrFun{\domain}{\thectt}
	{\GroundChi(z_1)\GroundEta(w_1)\cdots
	\Ground(x_{k+1})\GroundPartner(x_k)
	\cdots\GroundPartner(x_1)}
	=
	\BigCorrFun{\domain}{\thectt}
	{\GroundChi(z_1)\GroundEta(w_1)\cdots\GroundPartner(x_k)\cdots\GroundPartner(x_1)}\,.
\end{align}
Note that this argument is valid for $k=0$ and proves $\bigCorrFun{\domain}{\thectt}{\Ground(z)}=1$ for any $z\in\domain$.
Then, on the one hand, we can use property \textbf{(LOG)} which dictates that, as $\vert\zeta-x_{k+1}\vert\to0$, 
\begin{align}\label{eq: recursion partner CF}
		\BigCorrFun{\domain}{\thectt}
	{\GroundChi(z_1)\GroundEta(w_1)\cdots&\,\GroundChi(\zeta)\GroundEta(x_{k+1})
		\cdots\GroundPartner(x_1)}
	=
	\\\nonumber
	& -\,\BigCorrFun{\domain}{\thectt}
	{\GroundChi(z_1)\GroundEta(w_1)\cdots\,\GroundPartner(x_{k+1})
		\cdots\GroundPartner(x_1)}
	\\\nonumber
	& +\,
	\bigg(\log\frac{1}{\vert\zeta-x_{k+1}\vert^2}-\thectt\bigg)
	\BigCorrFun{\domain}{\thectt}
	{\GroundChi(z_1)\GroundEta(w_1)\cdots
	\Ground(x_{k+1})\cdots\GroundPartner(x_1)}
	\\\nonumber
	& +\,
	\phantom{\bigg\vert}
	o\big(\vert \zeta - x_{k+1}\vert\big)\,.
\end{align}
On the other hand, by induction hypothesis we can rewrite
\begin{align}\label{eq: rewriting}
\BigCorrFun{\domain}{\thectt}
{&\GroundChi(z_1)\GroundEta(w_1)\cdots\GroundChi(\zeta)\GroundEta(x_{k+1})
	\cdots\GroundPartner(x_1)}
\phantom{\Bigg\vert}
\\\nonumber
\phantom{\Bigg\vert}
&
=
\BigCorrFun{\domain}{\thectt}
{\GroundChi(\zeta)\GroundEta(x_{k+1})}
(-1)^k
\sum_{\underset{b(\zeta)=x_{k+1}}{b\in\Bij(\mathbf{z}^+,\mathbf{w}^+;\mathbf{x})}}
(-1)^b(-1)^{F_b}
\prod_{\underset{b(y)\neq y}{y\in\mathbf{z}\cup\mathbf{x}}}
\BigCorrFun{\domain}{\thectt}
{\GroundChi(y)\GroundEta\big(b(y)\big)}
\prod_{\underset{b(x)= x}{x\in\mathbf{x}}}
\BigCorrFun{\domain}{\thectt}
{\GroundPartner(x)}
\\\nonumber
\phantom{\Bigg\vert}& \mspace{40mu} +\,
(-1)^k
\sum_{\underset{b(\zeta)\neq x_{k+1}}{b\in\Bij(\mathbf{z}^+,\mathbf{w}^+;\mathbf{x})}}
(-1)^b(-1)^{F_b}
\prod_{\underset{b(y)\neq y}{y\in\mathbf{z}\cup\mathbf{x}}}
\BigCorrFun{\domain}{\thectt}
{\GroundChi(y)\GroundEta\big(b(y)\big)}
\prod_{\underset{b(x)= x}{x\in\mathbf{x}}}
\BigCorrFun{\domain}{\thectt}
{\GroundPartner(x)}\,,
\end{align}
with $\mathbf{z}^+=\{z_1,\ldots,z_n,\zeta\}$ and $\mathbf{w}^+=\{w_1,\ldots,w_n,x_{k+1}\}$.
Note that the subset of bijections~${b\in\Bij(\mathbf{z}^+,\mathbf{w}^+;\mathbf{x})}$ satisfying $b(\zeta)=x_{k+1}$ ---that is, the ones that contribute to the first term in~\eqref{eq: rewriting}--- is naturally in one-to-one correspondence with the subset of bijections~$b^+\in\Bij(\mathbf{z},\mathbf{w};\mathbf{x}^+)$ satisfying $b^+(x_{k+1})=x_{k+1}$, where $\mathbf{x}^+\coloneqq\{x_1,\ldots,x_{k+1}\}$.
Moreover, the correspondence satisfies $(-1)^k(-1)^{F_b}=(-1)^{k+1}(-1)^{F_{b^+}}$.
Note, as well, that, in the second term in~\eqref{eq: rewriting}, one can take the limit $\zeta\to x_{k+1}$; it can hence be rewritten as a sum over the subset of bijections~$b^+\in\Bij(\mathbf{z},\mathbf{w};\mathbf{x}^+)$ satisfying $b^+(x_{k_1})\neq x_{k_1}$.
Combining~\eqref{eq: rewriting} and~\eqref{eq: recursion partner CF}, we get
\begin{align*}
\BigCorrFun{\domain}{\thectt}
{\GroundChi(z_1)&\GroundEta(w_1)\cdots\,\GroundPartner(x_{k+1})
	\cdots\GroundPartner(x_1)}
\phantom{\Bigg\vert}
+\,
\phantom{\bigg\vert}
o\big(\vert \zeta - x_{k+1}\vert\big)\,
\\\nonumber
=
&
\;
\Bigg( -\BigCorrFun{\domain}{\thectt}{\GroundChi(\zeta)\GroundEta(x_{k+1})}
+
\log\frac{1}{\vert\zeta-x_{k+1}\vert^2}
-
\thectt \Bigg)\times
\\\nonumber
&\mspace{80mu}
\times
\phantom{\Bigg\vert}
(-1)^{k+1}
\sum_{\underset{b(x_{k+1})=x_{k+1}}{b\in\Bij(\mathbf{z},\mathbf{w};\mathbf{x}^+)}}
(-1)^b(-1)^{F_b}
\prod_{\underset{b(y)\neq y}{y\in\mathbf{z}\cup\mathbf{x}}}
\BigCorrFun{\domain}{\thectt}
{\GroundChi(y)\GroundEta\big(b(y)\big)}
\prod_{\underset{b(x)= x}{x\in\mathbf{x}}}
\BigCorrFun{\domain}{\thectt}
{\GroundPartner(x)}
\\\nonumber
& -
(-1)^k
\sum_{\underset{b(\zeta)\neq x_{k+1}}{b\in\Bij(\mathbf{z},\mathbf{w};\mathbf{x}^+)}}
(-1)^b(-1)^{F_b}
\prod_{\underset{b(y)\neq y}{y\in\mathbf{z}\cup\mathbf{x}^+}}
\BigCorrFun{\domain}{\thectt}
{\GroundChi(y)\GroundEta\big(b(y)\big)}
\prod_{\underset{b(x)= x}{x\in\mathbf{x}^+}}
\BigCorrFun{\domain}{\thectt}
{\GroundPartner(x)}
\end{align*}
Note that the factor in parenthesis has a well-defined limit by the property~\textbf{(FER)} and the properties of the Green's function~\eqref{eq: Green functions}.
In particular, recalling that we have ${\bigCorrFun{\domain}{\thectt}{\Ground(x)}=1}$ for all $x\in\domain$, the property~\textbf{(LOG)} dictates
\begin{align}\label{eq: 1pt log partner}
	-\;
	\BigCorrFun{\domain}{\thectt}{\GroundChi(\zeta)\GroundEta(x)}
	+\;
	\log\frac{1}{\vert\zeta-x\vert^2}
	\,-
	\thectt
	\;\;\xrightarrow{\ \ \zeta\,\to\,x\ \ }\;\; \BigCorrFun{\domain}{\thectt}{\GroundPartner(x)}
	=
	-4\pi\HarmOfGreen_\domain(x,x)-\thectt
	\,,
\end{align}
concluding the proof of the expression~\eqref{eq: formula step 1}.
This recursive construction guarantees the uniqueness of such correlation functions.

\noindent\textit{Step 2}:
We can exploit the properties~\textbf{(DER)} and~\textbf{(CUR)} to recursively construct in a unique way the correlation functions of all fields from the correlation functions in Step~1.
First, the algebraic relations in Remark~\ref{rmk: currents derivatives of fermions} and~\textbf{(DER)} imply
\begin{align*}
	\phantom{\bigg\vert}
	\bigCorrFun{\domain}{\thectt}{\,\HolCurrentChi (z)\cdots}
	&
	=
	\partial_{z}\bigCorrFun{\domain}{\thectt}{\,\GroundChi(z)\cdots}\,,
	\mspace{70mu}
	\phantom{\bigg\vert}
	\bigCorrFun{\domain}{\thectt}{\,\AntiHolCurrentChi (z)\cdots}
	=
	\overline{\partial}_{z}\bigCorrFun{\domain}{\thectt}{\,\GroundChi(z)\cdots}\,,
	\\
	\phantom{\bigg\vert}
	\bigCorrFun{\domain}{\thectt}{\,\HolCurrentEta (z)\cdots}
	&
	=
	\partial_{z}\bigCorrFun{\domain}{\thectt}{\,\GroundEta(z)\cdots}\,,
	\mspace{15mu}
	\text{ and }
	\mspace{15mu}
	\phantom{\bigg\vert}
	\bigCorrFun{\domain}{\thectt}{\,\AntiHolCurrentEta (z)\cdots}
	=
	\overline{\partial}_{z}\bigCorrFun{\domain}{\thectt}{\,\GroundEta(z)\cdots}\,.
\end{align*}
Then, since the state~$\GroundPartner$ is cyclic in~$\FullFock$, we can obtain the correlation functions of any field by repeated extraction of OPE coefficients with the currents.
That is, given any field~$\field\in\FullFock$, we define
\begin{align}\label{eq: CF action chi}
\BigCorrFun{\domain}{\thectt}{\cdots \big[\AlgChi{k} \field \big](z) \cdots}
\;\coloneqq&\;\phantom{-}\;
\oint_\gamma\frac{\dd\zeta}{2\pi\ii}(\zeta-z)^k\,
\BigCorrFun{\domain}{\thectt}{\cdots \HolCurrentChi(\zeta)\field(z) \cdots}\,,
\\\label{eq: CF action eta}
\BigCorrFun{\domain}{\thectt}{\cdots \big[\AlgEta{k} \field \big](z) \cdots}
\;\coloneqq&\;\phantom{-}\;
\oint_\gamma\frac{\dd\zeta}{2\pi\ii}(\zeta-z)^k\,
\BigCorrFun{\domain}{\thectt}{\cdots \HolCurrentEta(\zeta)\field(z) \cdots}\,,
\\\label{eq: CF action chi bar}
\BigCorrFun{\domain}{\thectt}{\cdots \big[\AlgChiBar{k} \field \big](z) \cdots}
\;\coloneqq&\;-
\oint_\gamma\frac{\dd\overline\zeta}{2\pi\ii}(\overline\zeta-\overline z)^k\,
\BigCorrFun{\domain}{\thectt}{\cdots \AntiHolCurrentChi(\zeta)\field(z) \cdots}\,,
\\\label{eq: CF action eta bar}
\text{ and }\;\;\;
\BigCorrFun{\domain}{\thectt}{\cdots \big[\AlgEtaBar{k} \field \big](z) \cdots}
\;\coloneqq&\;-
\oint_\gamma\frac{\dd\overline\zeta}{2\pi\ii}(\overline\zeta-\overline z)^k\,
\BigCorrFun{\domain}{\thectt}{\cdots \AntiHolCurrentEta(\zeta)\field(z) \cdots}\,,
\end{align}
where $\gamma$ is any positively-oriented contour small enough not to enclose any of the other field insertions.

To ensure the well-posedness of~\eqref{eq: CF action chi}~--~\eqref{eq: CF action eta bar}, we check that the current modes satisfy the appropriate relations.
First of all, using formula~\eqref{eq: formula step 1}, it is straightforward to check
\begin{align*}
	\BigCorrFun{\domain}{\thectt}{\cdots \big[\AlgEta{k}\AlgChi{\ell} \field \big](z) \cdots}
	+
	\BigCorrFun{\domain}{\thectt}{\cdots \big[ \AlgChi{\ell}\AlgEta{k} \field \big](z) \cdots}
	\,=\,
	k\,\delta_{k+\ell}
	\BigCorrFun{\domain}{\thectt}{\cdots \field(z) \cdots}\,.	
\end{align*}
Similarly, one can check that the current modes satisfy the anticommutation relations of two anticommuting copies of the symplectic fermions algebra.
In other words, the action defined by~\eqref{eq: CF action chi}~--~\eqref{eq: CF action eta bar} factors, at least, through $\FullEnvSymFer$.
It remains to check that the action factors through the quotient that defines the logarithmic Fock space~$\FullFock$.
Using formula~\eqref{eq: formula step 1}, we get
\begin{align}\label{eq: single current action}
\BigCorrFun{\domain}{\thectt}{\cdots \big[\AlgChi{k} \GroundPartner \big](z) \cdots}
\,=&\;
\oint_\gamma\frac{\dd\zeta}{2\pi\ii}(\zeta-z)^k\,
\BigCorrFun{\domain}{\thectt}{\cdots \HolCurrentChi(\zeta)\GroundPartner(z) \cdots}
\\\nonumber
=&\;
\oint_\gamma\frac{\dd\zeta}{2\pi\ii}(\zeta-z)^k\,
\partial_\zeta\Bigg(
\BigCorrFun{\domain}{\thectt}{\GroundChi(\zeta) \GroundEta(z)}
\BigCorrFun{\domain}{\thectt}{\cdots\GroundChi(z)\cdots}
\\\nonumber
&
\mspace{180mu}
\phantom{\Bigg\vert}
\;+\;
\mspace{10mu}
\text{ harmonic } \text{ w.r.t. }\; \zeta\; \text{ at }\; z\;\;
\Bigg)
\\\nonumber
=&\;-\BigCorrFun{\domain}{\thectt}{\cdots\GroundChi(z)\cdots}
\oint_\gamma\frac{\dd\zeta}{2\pi\ii}(\zeta-z)^{k-1}\phantom{\Bigg\vert}
\\\nonumber
&
\mspace{180mu}
\phantom{\Bigg\vert}
\;+\;\;
\oint_\gamma\frac{\dd\zeta}{2\pi\ii}(\zeta-z)^{k}
\Big(\text{holom. w.r.t $\zeta$ at $z$}\Big)\,,
\end{align}
and similar expressions for the action of the other current modes on $\GroundPartner$.
It is clear from~\eqref{eq: single current action} that the action factors through the relations ${\AlgChi{k} \GroundPartner = \AlgEta{k} \GroundPartner = \AlgChiBar{k} \GroundPartner = \AlgEtaBar{k} \GroundPartner = 0}$ for $k>0$.
By the anticommutation relations, to check the remaining relations~$\AlgChi{0}=\AlgChiBar{0}$ and $\AlgEta{0}=\AlgEtaBar{0}$ it suffices to consider them on the ground state~$\GroundPartner$.
From~\eqref{eq: single current action}, and by the same computation for $\AlgChiBar{0}$, we get
\begin{align}\label{eq: consistency}
\BigCorrFun{\domain}{\thectt}{\cdots \big[\AlgChi{0} \GroundPartner \big](z) \cdots}
\,
=\;
-\,\BigCorrFun{\domain}{\thectt}{\cdots\GroundChi(z)\cdots}
\,
=\;
\BigCorrFun{\domain}{\thectt}{\cdots \big[\AlgChiBar{0} \GroundPartner \big](z) \cdots}\,.
\end{align}
Similar computations lead to $\AlgEta{0}\GroundPartner=\AlgEtaBar{0}\GroundPartner$ within correlation functions.
We conclude hence that the action defined by~\eqref{eq: CF action chi}~--~\eqref{eq: CF action eta bar} factors through the quotient~$\FullFock$.

\noindent
\textit{Step 3}:
It remains to check the properties~\textbf{(FER)}, \textbf{(DER)},
\textbf{(CUR)}, and \textbf{(LOG)}.

\noindent
\textbf{(FER)}:
In terms of the basis~\eqref{eq: full basis}, the ground fermions are written as~${\GroundChi=-\AlgChi{0}\GroundPartner}$ and \linebreak ${\GroundEta=-\AlgEta{0}\GroundPartner}$.
From~\eqref{eq: consistency} and its equivalent for the \mbox{$\GroundEta$-currents}, we see that those relations are satisfied inside correlation functions.

\noindent
\textbf{(DER)}:
It is enough to prove it for fields in the basis~\eqref{eq: full basis}
\begin{align*}
	\field
	\coloneqq
	\AlgEta{-k_r}\cdots\AlgEta{-k_2}\AlgEta{-k_1}
	\AlgChi{-\ell_s}\cdots\AlgChi{-\ell_2}\AlgChi{-\ell_1}
	\AlgEtaBar{-\bar{k}_{\bar{r}}}\cdots\AlgEtaBar{-\bar{k}_2}\AlgEtaBar{-\bar{k}_1}
	\AlgChiBar{-\bar{\ell}_{\bar{s}}}\cdots\AlgChiBar{-\bar{\ell}_2}\AlgChiBar{-\bar{\ell}_1}
	\GroundPartner\,,
\end{align*}
and we proceed inductively on $r$, $s$, $\bar{r}$, and $\bar{s}$.
We spell out here the induction on $r$.
For the sake of brevity, we check the base case $\field=\GroundPartner$ when no other fields are inserted.
Recall that, by the Sugawara construction, we have $\SugL{-1} \GroundPartner = \AlgEta{-1}\GroundChi-\AlgChi{-1}\GroundEta$
\begin{align*}
	\BigCorrFun{\domain}{\thectt}{ \big[ \SugL{-1} \GroundPartner \big](z) }
	\;&=\;
	\BigCorrFun{\domain}{\thectt}{ \big[ \AlgEta{-1}\GroundChi \big](z) }
	-
	\BigCorrFun{\domain}{\thectt}{ \big[ \AlgChi{-1}\GroundEta \big](z) }
	\\
	&=\;
	\oint_\gamma \frac{\dd\zeta}{2\pi\ii}\,
	\frac{1}{\zeta - z}
	\Bigg(
	\BigCorrFun{\domain}{\thectt}{ \HolCurrentEta(\zeta)\GroundChi (z) }
	-
	\BigCorrFun{\domain}{\thectt}{ \HolCurrentChi(\zeta)\GroundEta (z) }
	\Bigg)
	\\
	&=\;
	-4\pi
	\oint_\gamma \frac{\dd\zeta}{2\pi\ii}\,
	\frac{1}{\zeta - z}
	\Bigg(
	\partial_{(1)} \Green_\domain(\zeta,z)
	+
	\partial_{(2)} \Green_\domain(z,\zeta)
	\Bigg)
	\\
	&=\;
	-4\pi
	\Big(
	\partial_{(1)} \HarmOfGreen_\domain(z,z)
	+
	\partial_{(2)} \HarmOfGreen_\domain(z,z)
	\Big)
	\phantom{\bigg\vert}
	\\
	&=\;
	\partial_{z} \BigCorrFun{\domain}{\thectt}{ \GroundPartner(z) }\,,
\end{align*}
where in the last equality we have used the explicit expression~\eqref{eq: 1pt log partner}.
For the general base case it is enough to consider other insertions of $\GroundChi$, $\GroundEta$ and $\GroundPartner$ since any other correlation function can be extracted by contour integration around points away from $z$.
This can be handled using~\eqref{eq: formula step 1}.
For the induction step, we have
\begin{align*}
	\BigCorrFun{\domain}{\thectt}{\cdots 
		\big[ \SugL{-1} & \AlgEta{-k_{r+1}} \field \big](z)
		\cdots}
	\\
	&=\;
	\BigCorrFun{\domain}{\thectt}{\cdots
		\big[ \AlgEta{-k_{r+1}} \SugL{-1} \field \big](z)
		\cdots}
	\phantom{\Bigg\vert}
	+\;
	\BigCorrFun{\domain}{\thectt}{\cdots
		\Big[ [ \SugL{-1} , \AlgEta{-k_{r+1}} ] \field \Big](z)
		\cdots}
	\\
	&=\;
	\oint_\gamma \frac{\dd\zeta}{2\pi\ii}
	\frac{\bigCorrFun{\domain}{\thectt}{\cdots
			\HolCurrentEta(\zeta) \big[ \SugL{-1} \field \big](z)
			\cdots}}{(\zeta - z)^{k_{r+1}}}
	\phantom{\Bigg\vert}
	+\;
	k_{r+1}
	\BigCorrFun{\domain}{\thectt}{\cdots
		\big[  \AlgEta{-k_{r+1}-1} \field \big](z)
		\cdots}\,,
\end{align*}
where $\gamma$ is a small enough contour surrounding $z$.
By induction hypothesis, the contour integral can be expressed as
\begin{align*}
	\oint_\gamma \frac{\dd\zeta}{2\pi\ii} &\;
	\frac{\partial_{z}\bigCorrFun{\domain}{\thectt}{\cdots
			\HolCurrentEta(\zeta) \field(z)
			\cdots}}{(\zeta - z)^{k_{r+1}}}
	\\
	&
	\mspace{50mu}
	=\;
	\oint_\gamma \frac{\dd\zeta}{2\pi\ii}
	\Bigg[\,
	\partial_{z}\Bigg(
	\frac{\bigCorrFun{\domain}{\thectt}{\cdots
			\HolCurrentEta(\zeta) \field(z)
			\cdots}}{(\zeta - z)^{k_{r+1}}}
		\Bigg)
	-\;
	k_{r+1}
	\frac{\bigCorrFun{\domain}{\thectt}{\cdots
			\HolCurrentEta(\zeta) \field(z)
			\cdots}}{(\zeta - z)^{k_{r+1}+1}}
	\,\Bigg]
	\\
	&
	\mspace{50mu}
	=\;
	\partial_{z}\BigCorrFun{\domain}{\thectt}{\cdots
			\big[\AlgEta{-k_{r+1}} \field \big] (z)
			\cdots}
	-\;
	k_{r+1}
	\BigCorrFun{\domain}{\thectt}{\cdots
		\big[\AlgEta{-(k_{r+1}+1)} \field \big] (z)
		\cdots}\,,
	\phantom{\Bigg\vert}
\end{align*}
which concludes the proof for the holomorphic mode~$\AlgL{-1}$.
The proof for $\AlgLBar{-1}$ follows the same lines.

\noindent
\textbf{(CUR)}:
The correlation function
\begin{align*}
	\BigCorrFun{\domain}{\thectt}{\field_1(z_1)
		\cdots \field_n(z_n)
	\HolCurrentChi(z)\field(w)}
\end{align*}
is holomorphic with respect to~$z$ away from the other insertions~$z_1,\ldots,z_n$ and $w$ by construction ---Steps~1 and~2---.
Hence, it has a Laurent expansion around $w$, which, by construction ---\eqref{eq: CF action chi}--\eqref{eq: CF action eta bar}---, must be
\begin{align*}
	\BigCorrFun{\domain}{\thectt}{\field_1(z_1)
		\cdots \field_n(z_n)
		\HolCurrentChi(z)\field(w)}
	=\;
	\sum_{k\in\Z}\,
	\frac{\bigCorrFun{\domain}{\thectt}{\field_1(z_1)
			\cdots \field_n(z_n)\,
			\big[\AlgChi{-k}\field\big](w)\,}}
		{(z-w)^{k+1}}\,.
\end{align*}
The same argument works for the currents $\AntiHolCurrentChi$, $\HolCurrentEta$ and $\AntiHolCurrentEta$.

\noindent
\textbf{(LOG)}:
It is enough to consider the correlation function
\begin{align*}
	\BigCorrFun{\domain}{\thectt}
	{\GroundChi(z)\GroundEta(w)\GroundChi(z_1)\GroundEta(w_1)\cdots\GroundChi(z_n)\GroundEta(w_n)
		\GroundPartner(x_k)\cdots\GroundPartner(x_1)}
\end{align*}
from~\eqref{eq: formula step 1}, since a more general one is obtained by differentiation and residue extraction away from $z$ and $w$.
Let us write the correlation function at hand as
\begin{align}\label{eq: OPE fermions}
\BigCorrFun{\domain}{\thectt}
{\GroundChi(z)\GroundEta(w)\,\GroundChi(z_1)&\GroundEta(w_1)\cdots\GroundChi(z_n)\GroundEta(w_n)
	\GroundPartner(x_k)\cdots\GroundPartner(x_1)}\phantom{\bigg\vert}\\ \nonumber
&=\phantom{\bigg\vert}
\BigCorrFun{\domain}{\thectt}
{\GroundChi(z)\GroundEta(w)}
\BigCorrFun{\domain}{\thectt}{
\GroundChi(z_1)\GroundEta(w_1)\cdots\GroundChi(z_n)\GroundEta(w_n)
	\GroundPartner(x_k)\cdots\GroundPartner(x_1)}
+
\; (\,\star\,)\,,
\end{align}
where the term $(\,\star\,)$ contains the rest of terms in~\eqref{eq: formula step 1}.
Let us already make the observation that, with a simple combinatorial argument, from~\eqref{eq: formula step 1} we get
\begin{align}\label{eq: combinatorial fact}
\BigCorrFun{\domain}{\thectt}
{\GroundPartner(w)\,\GroundChi(z_1)\GroundEta(w_1)&\cdots\GroundChi(z_n)\GroundEta(w_n)
	\GroundPartner(x_k)\cdots\GroundPartner(x_1)}\phantom{\bigg\vert}
\\\nonumber
&=\phantom{\bigg\vert}
\BigCorrFun{\domain}{\thectt}
{\GroundPartner(w)}
\BigCorrFun{\domain}{\thectt}{
	\GroundChi(z_1)\GroundEta(w_1)\cdots\GroundChi(z_n)\GroundEta(w_n)
	\GroundPartner(x_k)\cdots\GroundPartner(x_1)}
-
\; (\,\star\,)\,.
\end{align}
The two-point function can be expanded around $w$ as
\begin{align*}
	\BigCorrFun{\domain}{\thectt}
	{\GroundChi(z)\GroundEta(w)}
	\;=\;&
	\,4\pi\, \Green_\domain(z,w)\\
	=\;&\log\frac{1}{\vert z-w\vert^2}
	\,+\,
	4\pi\,
	\HarmOfGreen_\domain(z,w)\\
	=\;&\log\frac{1}{\vert z-w\vert^2}
	\,+\,
	4\pi\,
	\HarmOfGreen_\domain(w,w)
	\,+\, o\big(\vert z-w\vert\big)\,.	
\end{align*}
Then, using this expansion, the identities~\eqref{eq: OPE fermions} and~\eqref{eq: combinatorial fact} and the explicit one-point function~\eqref{eq: 1pt log partner}, we get
\begin{align*}
\BigCorrFun{\domain}{\thectt}
{\GroundChi(z)\GroundEta(w)\,&\GroundChi(z_1)\GroundEta(w_1)\cdots\GroundChi(z_n)\GroundEta(w_n)
	\GroundPartner(x_k)\cdots\GroundPartner(x_1)}\phantom{\bigg\vert}\\ \nonumber
&=\phantom{\bigg\vert}
\bigg(\log\frac{1}{\vert z-w\vert^2}
-\thectt
\bigg)
\BigCorrFun{\domain}{\thectt}{
	\GroundChi(z_1)\GroundEta(w_1)\cdots\GroundChi(z_n)\GroundEta(w_n)
	\GroundPartner(x_k)\cdots\GroundPartner(x_1)}
\\
&
\mspace{50mu}
-\;
\BigCorrFun{\domain}{\thectt}
{\GroundPartner(w)\,\GroundChi(z_1)\GroundEta(w_1)\cdots\GroundChi(z_n)\GroundEta(w_n)
	\GroundPartner(x_k)\cdots\GroundPartner(x_1)}\phantom{\bigg\vert}
+\;\,
o\big(\vert z - w\vert\big)\,.
\end{align*}
The proof of \textbf{(LOG)} is complete by recalling that we have $\bigCorrFun{\domain}{\thectt}{\Ground(z)\cdots}=\bigCorrFun{\domain}{\thectt}{\cdots}$.
\end{proof}

\newpage

\titleformat{\section}
{\normalfont\Large\bfseries}{\thesection}{0pt}{}

\end{document}